\def\year{2021}\relax
\documentclass[letterpaper]{article} % DO NOT CHANGE THIS
\usepackage{aaai21}  % DO NOT CHANGE THIS
\usepackage{times}  % DO NOT CHANGE THIS
\usepackage{helvet} % DO NOT CHANGE THIS
\usepackage{courier}  % DO NOT CHANGE THIS
\usepackage[hyphens]{url}  % DO NOT CHANGE THIS
\usepackage{graphicx} % DO NOT CHANGE THIS
\usepackage{natbib}  % DO NOT CHANGE THIS AND DO NOT ADD ANY OPTIONS TO IT
\usepackage{caption} % DO NOT CHANGE THIS AND DO NOT ADD ANY OPTIONS TO IT
\usepackage[utf8]{inputenc} % allow utf-8 input
\usepackage[T1]{fontenc}    % use 8-bit T1 fonts
\usepackage[hidelinks]{hyperref}       % hyperlinks
\usepackage{url}            % simple URL typesetting
\usepackage{booktabs}       % professional-quality tables
\usepackage{amsfonts}       % blackboard math symbols
\usepackage{nicefrac}       % compact symbols for 1/2, etc.
\usepackage{microtype}      % microtypography
\usepackage{times}
\usepackage{amsmath}
\usepackage{amssymb}
\usepackage{amsthm}
\usepackage{microtype}
\usepackage{enumitem}
\usepackage{colortbl}
\usepackage[capitalise,noabbrev]{cleveref}
\usepackage{xcolor}
\usepackage{etoolbox}
\usepackage{thm-restate}
\usepackage[linesnumbered,ruled,lined,noend]{algorithm2e}
\usepackage{amsmath}
\usepackage{mathtools}
\usepackage{tcolorbox}
\usepackage{amsthm}
\usepackage{mleftright}
\usepackage{stmaryrd}
\usepackage{float}
\usepackage{bm}
\usepackage{numprint}
\usepackage{environ}
\usepackage{tikz}
\usepackage{cleveref}
\usepackage{standalone}

\usepgflibrary{shapes.geometric,arrows.meta}

\ifcsmacro{assumption}{}{
	
}
\ifcsmacro{corollary}{}{
	\newtheorem{corollary}{Corollary}[]
}
\ifcsmacro{definition}{}{
	
}
\ifcsmacro{example}{}{
	
}
\ifcsmacro{fact}{}{
	
}
\ifcsmacro{informal_theorem}{}{
	
}
\ifcsmacro{lemma}{}{
	\newtheorem{lemma}{Lemma}[]
}
\ifcsmacro{proposition}{}{
	\newtheorem{proposition}{Proposition}[]
}
\ifcsmacro{remark}{}{
	\newtheorem{remark}{Remark}[]
}
\ifcsmacro{theorem}{}{
	\newtheorem{theorem}{Theorem}[]
}
\ifcsmacro{observation}{}{
	\newtheorem{observation}{Observation}[]
}
\mathtoolsset{centercolon}

\makeatletter
\patchcmd\algocf@Vline{\vrule}{\vrule \kern-0.4pt}{}{}
\patchcmd\algocf@Vsline{\vrule}{\vrule \kern-0.4pt}{}{}
\makeatother

\SetKwComment{Hline}{}{\vspace{-3mm}\textcolor{gray}{\hrule}\vspace{1mm}}
\definecolor{darkgrey}{gray}{0.3}
\definecolor{commentcolor}{gray}{0.5}
\SetKwComment{Comment}{\color{commentcolor}$\triangleright$\ }{}
\SetCommentSty{}
\SetNlSty{}{\color{darkgrey}}{}
\crefname{algocf}{Algorithm}{Algorithms}
\SetKwProg{Fn}{function}{}{}
\SetKwProg{Subr}{subroutine}{}{}

\newcommand{\linkfootnote}[1]{%
    \protect\hyperref[#1]{%
        \hbox{\textsuperscript{\normalfont\ref*{#1}}}%
    }%
}

\newcommand{\colorrule}[1]{%
    \arrayrulecolor{#1}\midrule\arrayrulecolor{black}%
}

\newcommand{\defeq}{\mathrel{:\mkern-0.25mu=}}

\newcommand{\cX}{\mathcal{X}}

\newcommand{\cJ}{\mathcal{J}}
\newcommand{\cK}{\mathcal{K}}

\newcommand{\cR}{\mathcal{R}}

\newcommand{\bbR}{\mathbb{R}}
\newcommand{\bbP}{\mathbb{P}}
\newcommand{\bbE}{\mathbb{E}}
\newcommand{\R}{\mathbb{R}}

\renewcommand{\vec}[1]{\bm{#1}}

\newcommand{\emptyseq}{\varnothing}

\newcommand{\Rp}{\bbR_{\ge 0}}

\title{Model-Free Online Learning in Unknown Sequential Decision Making Problems and Games}
\author{
Gabriele Farina,\textsuperscript{\rm 1}
Tuomas Sandholm\textsuperscript{\rm 1,\rm 2,\rm 3,\rm 4}\\
}
\affiliations {
\textsuperscript{\rm 1}Computer Science Department, Carnegie Mellon University,
\textsuperscript{\rm 2}Strategic Machine, Inc.,\\
\textsuperscript{\rm 3}Strategy Robot, Inc.,
\textsuperscript{\rm 4}Optimized Markets, Inc.\\
gfarina@cs.cmu.edu, sandholm@cs.cmu.edu
}
\makeatletter
\g@addto@macro \normalsize {%
 \addtolength\abovedisplayskip{-3pt}%
 \addtolength\belowdisplayskip{-3pt}%
}
\makeatother

\LetLtxMacro{\baseproof}{\proof}
\LetLtxMacro{\endbaseproof}{\endproof}

\tcbuselibrary{skins,breakable}
\newcommand{\runinsec}[1]{\noindent\textbf{#1}\quad}
\newcommand{\pihat}{\hat{\vec{\pi}}}

\newcommand{\bbone}{\bm{1}}

\allowdisplaybreaks[1]
\newcommand\numberthis[1]{\addtocounter{equation}{1}\tag{\theequation}\label{#1}}

\usepackage[switch]{lineno}
\makeatletter
    \AtBeginDocument{%
      \@ifpackageloaded{amsmath}{%
        \newcommand*\patchAmsMathEnvironmentForLineno[1]{%
          \expandafter\let\csname old#1\expandafter\endcsname\csname #1\endcsname
          \expandafter\let\csname oldend#1\expandafter\endcsname\csname end#1\endcsname
          \renewenvironment{#1}%
                           {\linenomath\csname old#1\endcsname}%
                           {\csname oldend#1\endcsname\endlinenomath}%
        }%
        \newcommand*\patchBothAmsMathEnvironmentsForLineno[1]{%
          \patchAmsMathEnvironmentForLineno{#1}%
          \patchAmsMathEnvironmentForLineno{#1*}%
        }%
        \patchBothAmsMathEnvironmentsForLineno{equation}%
        \patchBothAmsMathEnvironmentsForLineno{align}%
        \patchBothAmsMathEnvironmentsForLineno{flalign}%
        \patchBothAmsMathEnvironmentsForLineno{alignat}%
        \patchBothAmsMathEnvironmentsForLineno{gather}%
        \patchBothAmsMathEnvironmentsForLineno{multline}%
      }{}
    }
\makeatother
\begin{document}
    %\linenumbers
    \maketitle

    \begin{abstract}
        Regret minimization has proved to be a versatile tool for tree-form sequential decision making and extensive-form games. In large two-player zero-sum imperfect-information games, modern extensions of counterfactual regret minimization (CFR) are currently the practical state of the art for computing a Nash equilibrium. Most regret-minimization algorithms for tree-form sequential decision making, including CFR, require
        (i) an exact \emph{model} of the player's decision nodes, observation nodes, and how they are linked, and
        (ii) full knowledge, at all times $t$, about the payoffs---even in parts of the decision space that are not encountered at time $t$.
        Recently, there has been growing interest towards relaxing some of those restrictions and making regret minimization applicable to settings for which reinforcement learning methods have traditionally been used---for example, those in which only black-box access to the environment is available. We give the first, to our knowledge, regret-minimization algorithm that guarantees sublinear regret with high probability even when requirement (i)---and thus also (ii)---is dropped. We formalize an online learning setting in which the strategy space is not known to the agent and gets revealed incrementally whenever the agent encounters new decision points. We give an efficient algorithm that achieves $O(T^{3/4})$ regret with high probability for that setting, even when the agent faces an adversarial environment. Our experiments show it significantly outperforms the prior algorithms for the problem, which do not have such guarantees.
    %It can be used for multiagent reinforcement learning.
    It can be used in any application for which regret minimization is useful: approximating Nash equilibrium or quantal response equilibrium, approximating coarse correlated equilibrium in multi-player games, learning a best response, learning safe opponent exploitation, and online play against an unknown opponent/environment.
    \end{abstract}

    \section{Introduction}

%\todo{People seem to like the term ``no domain knowledge''}

A \emph{tree-form sequential decision making (TFSDM)} problem
%That is not a standard term so I am not a huge fan of stating it as a standard term, especially right up front. ???
formalizes in a tree-form structure the interaction of an agent with an unknown and potentially adversarial environment. The agent's tree includes decision nodes, observation nodes, and terminal nodes. TFSDM captures the problem that a player faces in an extensive-form game. TFSDM also captures MDPs and POMDPs where the agent conditions on observed history, but TFSDM problems are more general because the Markovian assumption is not necessarily made.

In TFSDM, the environment may react adversarially to what the agent does. This is important to take into account, for example, in game-theoretic settings and \emph{multiagent reinforcement learning (MARL)} because the other agents' learning makes the environment nonstationary for the agent~\citep{Sandholm96:Multiagent,Matignon12:Independent}. This is in contrast to the standard assumption in single-agent reinforcement learning where the environment is \emph{oblivious} to the agent instead of adversarial. Hence, learning strong \emph{policies} (aka. \emph{strategies}) in TFSDM problems is especially challenging, and the agent must be careful about exploration because exploration actions can change the environment.

Online, regret minimization methods have been successfully used in TFSDM. In particular, the \emph{counterfactual regret minimization (CFR)} framework decomposes the overall regret of an agent to local regrets at individual decision nodes (aka. information sets in game theory)~\cite{Zinkevich07:Regret}. That enables significantly larger TFSDMs to be tackled. Many enhancements have been developed on top of the basic CFR framework~\cite{Lanctot09:Monte,Tammelin14:Solving,Brown15:Regret,Brown17:Dynamic,Brown17:Reduced,Brown19:Solving,Farina19:Online}, and have led to major milestones in im\-per\-fect-information games such as poker~\cite{Bowling15:Heads,Moravvcik17:DeepStack,Brown17:Superhuman,Brown19:Superhuman}.
Many of those methods \emph{guarantee} low regret even against an adversarial environment---which, in turn, enables the computation of game-theoretic solutions such as Nash equilibrium, coarse correlated equilibrium~\citep{Moulin78:Strategically,Celli19:Learning}, best responses, etc.

However, those methods usually come with two drawbacks:
(i) they require an explicit upfront \emph{model} of the agent's decision space, and
(ii) depending on the online learning model used, they require full feedback, at all times $t$, about the payoffs assigned by the environment---even in parts of the decision space not encountered at time $t$.
There has been work towards an online learning setting, called the \emph{bandit optimization setting}, that drops (ii)~\cite{Lattimore20:Bandit}. Most MARL algorithms apply for the unknown game setting and drop both (i) and (ii), often %ALWAYS???
at the cost of invalidating any regret guarantees. %~\cite{CHECK TO SEE WHETHER MY AWESOME PAPER OR BL-WOLF PAPER WOULD HAVE SOMETHING RELEVANT HERE, READ_AND CITE MY TEXTBOOK_ARTICLE AND GET ADDITOINAL CITATIONS FROM IT???}.
In this paper, we give, to our knowledge, the first regret-minimization algorithm that guarantees sublinear (specifically $O(T^{3/4} \sqrt{\log 1/p})$ with probability $1-p$) regret even when requirements (i) and (ii) are both dropped.

Conceptually, our algorithm has elements of both online bandit optimization and MARL.
On the one hand, our regret guarantees hold with high probability no matter how the environment picks its actions or assigns utilities to terminal nodes at all iterations $t$. In fact, our algorithm is a regret minimizer in the proper online-learning sense: it does not need to know---and makes no assumptions about---the underlying policy of the environment or the utility at the terminal states. Furthermore, those quantities are not assumed to be time independent and they can even be selected adversarially by the environment at each iteration. This is contrast with many \emph{self-play} methods, that require control over the opponent in order to retain guarantees. In particular, because we assume no control over the adversarial environment, every interaction with the environment can lead it to react and change behavior in the next iteration. So, it is impossible to ``freeze'' the policy of the environment to perform off-policy exploration like many self-play methods require. 
Because of its strong online guarantees, our algorithm can be used for all the applications in which regret minimization provides benefits---for example, to converge to Nash equilibrium in a two-player zero-sum extensive-form game, to learn a best response against static opponents, to converge to coarse correlated equilibrium in a multiagent setting~\cite{Moulin78:Strategically,Celli19:Learning}, to converge to a quantal-response equilibrium~\cite{Ling18:What,Farina19:Online}, to compute safe exploitative strategies~\citep{Farina19:Online,Ponsen11:Computing}, or to play online against an unknown opponent/environment.

On the other hand, ours is \emph{not} an algorithm for the online bandit optimization problem. In bandit optimization, the agent does \emph{not} interact with the environment: at all times $t$, the agent outputs a policy $\vec{\pi}^t$ for the whole decision space, and receives a single real number $u^t$, representing the gain incurred by $\vec{\pi}^t$, as feedback. Instead, our algorithm operates within a slightly different online learning model that we introduce, which we call the \emph{interactive bandit model}. In it, the decision maker gets to observe the signals (actions) selected by the environment on the path from the root of the decision problem to the agent's terminal state, in additions to $u^t$. Hence, we operate within a more specific %weaker
online learning model than bandit optimization (which applies to, for example, picking a point on a sphere also, as opposed to just TFSDM), and rather one that is natural to TFSDM. This comes with a significant advantage. While, to our knowledge, all algorithms for the bandit optimization problem require \emph{a priori} knowledge of the full structure of the tree-from sequential decision problem,\footnote{This is needed, for example, to construct a self-concordant barrier function~\cite{Abernethy08:Competing} or a global regularizer for the strategy space~\cite{Abernethy09:Beating,Farina20:Counterfactual}, to compute a barycentric spanner~\cite{Awerbuch04:Adaptive,Bartlett08:High}, or to set up a kernel function~\cite{Bubeck17:Kernel}.} our algorithm for the interactive bandits model is \emph{model free}. Here, the structure of the tree-form sequential decision process is at first unknown and has to be discovered by exploring as part of the decision-making process. Decision and observation nodes are revealed only at the time the agent encounters them for the first time.
%This brings our setup closer in spirit to MARL.
%
%The model-free setup leads to new challenges in online learning. For example, is our $O(T^{3/4})$ regret bound optimal when the SDM structure is unknown?
%Or does an algorithm with a stronger regret guarantees exist even in a model-free regimen?

%\runinsec{Relationship with MCCFR}
\subsection{Closely Related Research}
In the rest of this section we discuss how our algorithm relates to other attempts to connect online learning guarantees with model-free MARL. A comparison between our algorithm and other settings in online learning is deferred to \cref{sec:online learning}, where we formally introduce our interactive bandit model.

The greatest inspiration for our algorithm is the \emph{online variant of Monte Carlo CFR (online MCCFR)} algorithm proposed in passing by~\citep{Lanctot09:Monte}. Unlike the traditional ``self-play'' MCCFR, \emph{online} MCCFR does not assume that the algorithm have control over the environment.
The authors note that in theory online MCCFR could be used to play games in a model-free fashion, provided that a lower bound on the reach probability of every terminal state can be enforced. That lower bound, say $\eta$, is necessary for them to invoke their main theorem, which guarantees $O(\eta^{-1} T^{1/2} \sqrt{1/p})$ regret with probability $1-p$.
They suggest perhaps using some form of $\epsilon$-greedy exploration at each decision node to enforce the lower bound, but no guarantee is provided and the authors then move away from this side note to focus on the self-play case.
We show that their proposed approach encounters significant hurdles. First, using  exploration at each decision node results in a lower bound on the reach of every terminal state on the order of $\eta = \epsilon^d$, where $d$ is the depth of the decision process, thus making the regret bound not polynomial in the size of the decision process itself, but rather exponential. Second, the paper did not provide theoretical guarantees for the online case. In particular, the theory does not take into account the degredation effect caused by the exploration itself, which scales roughly as $\eta T$. So, on the one hand, a large $\eta$ is needed to keep the term $\eta^{-1}$ in their regret bound under control, but at the same time a large $\eta$ results in a $\eta T$ term being added to the regret. These hurdles show that it is unlikely that their approach can lead to $O(T^{1/2}\sqrt{1/p})$ regret with high probability $1-p$ as they hypothesized. We address those issues by using a different type of exploration and being careful about bounding the degradation terms in the regret incurred due to the exploration. Because of the latter, our algorithm incurs $O(T^{3/4}\sqrt{\log(1/p)})$ regret with high probability against adversarial opponents. Because the exponent is less than 1, ours is truly a regret minimization algorithm for TFSDMs with unknown structure, and to our knowledge, the first. At the same time, our exponent is worse than the originally hypothesized exponent 1/2. It is unknown whether the latter can be achieved. Finally, our dependence on $p$ is better than in their hypothesized regret bound.
%Also, they assumed that the number of training rounds, $t$, is known while we do not assume that.

A recent paper by~\citet{Srinivasan18:Actor} related policy gradient algorithms to CFR (and, to a much lesser degree, MCCFR). Despite their experiments demonstrating empirical convergence rates for sampled versions of their algorithms, formal guarantees are only obtained for tabularized policy iteration in self-play, and use policy representations that require costly $\ell_2$ projections back into the policy space. In contrast, our regret guarantees hold (i) in any TFSDM setting (not just two-player zero-sum extensive-form games), (ii) in high probability, (iii) with sampling, and (iv) when playing against any environment, even an adversarial one, without requiring complex projections.

A very recent paper~\citep{Zhou20:Posterior} on two-player zero-sum games proposes combining full-information regret minimization with posterior sampling~\citep{Strens00:Bayesian} to estimate the utility function of the player and transition model of chance, both of which they assume to be time independent unlike our setting. They show in-expectation bounds under the assumption that the agent's strategy space is known \emph{ex ante}.
We operate in a significantly more general setting where the observations and utilities are decided by the environment and can change---even adversarially---between iterations.
Like theirs, our algorithm converges to Nash equilibrium in two-player zero-sum games when used in self play. However, unlike theirs, our algorithm is a regret minimizer that can be used for other purposes also, such as finding a coarse correlated equilibrium in multiplayer general-sum games, a quantal-response equilibrium, or safe exploitative strategies. In the latter two applications, the payoff function, in effect, changes as the agent changes its strategy. Our regret guarantees hold in high probability and we do not assume \emph{ex ante} knowledge of the agent's strategy space.

A different line of work has combined fictitious play \cite{Brown51:Iterative}
%, which is known to converge to Nash equilibria in two-player zero-sum games,
with deep learning for function approximation~\citep{Heinrich15:Fictitious,Heinrich16:Deep}. Those methods do not give regret guarantees.
Finally, other work has studied how to combine the guarantees of online learning 
with MDPs. \citet{Kash19:Combining} combine the idea of breaking up and minimizing regret locally at each decision point, proper
%Is proper a typo???
of CFR, with Q-learning, obtaining an algorithm with certain in-the-limit guarantees for MDPs.
\citet{Even09:Online} study online optimization (in the full-feedback setting, as opposed to bandit) in general MDPs where the reward function and the structure of the MDP is known. \citet{Neu10:Online} study online bandit optimization in MDPs in the \emph{oblivious} setting, achieving $O(T^{2/3})$ regret with high probability, again assuming that the MDP's structure is known and certain conditions are met. \citet{Zimin13:Online} give bandit guarantees for episodic MDPs with a fixed known transition function.

    \section{Our Model for (Unknown) Tree-Form Sequential Decision Making and Games}
%\vspace{-3mm}

In this section, we introduce the notation for TFSDM problems that we will be using for the rest of the paper. 

A \emph{tree-form sequential decision making (TFSDM) problem} is structured as a tree made of three types of nodes: (i) \emph{decision nodes} $j$, in which the agent acts by selecting an action from the finite set $A_j$ (different decision nodes can admit different sets of allowed actions); (ii) \emph{observation points} $k$, where the agent observes one out of set $S_k$ of finitely many possible signal from the environment (different observation points can admit different sets of possible signals); and (iii) \emph{terminal nodes}, corresponding to the end of the decision process. We denote the set of decision nodes in the tree-form sequential decision making problem as $\cJ$, the set of observation points as $\cK$, and the set of terminal nodes as $Z$. Furthermore, we let $\rho$ denote the dynamics of the game: selecting action $a \in A_j$ at decision node $j\in\cJ$ makes the process transition to the next state $\rho(j,a) \in \cJ \cup \cK \cup Z$, while the process transitions to $\rho(k,s) \in \cJ \cup \cK \cup Z$ whenever the agent observes signal $s \in S_k$ at observation point $k \in \cK$.

%As we mentioned in the introduction, our algorithm (\cref{sec:algorithm}) operates under partial knowledge of the structure of the decision process. In particular, the structure of the

Our algorithm operates in the difficult setting where the structure of the TFSDM problem is at first unknown and can be discovered only through exploration. Decisions and observation nodes are revealed only at the time the agent encounters them for the first time. As soon as a decision node $j$ is revealed for the first time, its corresponding set of actions $A_j$ is revealed with it.

\runinsec{Sequences}
In line with the game theory literature, we call a \emph{sequence} a decision node-action pair; each sequence $(j,a)$ uniquely identifies a path from the root of the decision process down to action $a$ at decision node $j$, included. Formally, we define the set of sequences as $\Sigma \defeq \{(j, a): j\in\cJ, a \in A_{j}\} \cup \{\emptyseq\}$,  where the special element $\emptyseq$ is called the \emph{empty sequence}.
Given a decision node $j \in \cJ$, its \emph{parent sequence}, denoted $p_j$, is the last sequence (that is, decision node-action pair) encountered on the path from the root of the decision process down to $j$. If the agent does not act before $j$ (that is, only observation points are encountered on the path from the root to $j$), we let $p_j = \emptyseq$.

Given a terminal node $z \in Z$ and a sequence $(j, a) \in \Sigma$, we write $(j,a) \leadsto z$ to mean that the path from the root of the decision process to $z$ passes through action $a$ at decision node $j$. Similarly, given a terminal node $z \in Z$ and an observation node-signal pair $(k, s)$ ($s \in S_k$), we write $(k,s) \leadsto z$ to indicate that the path from the root of the decision process to $z$ passes through signal $s$ at observation node $k$. Finally, we let $\sigma(z)$ be the last sequence (decision node-action pair) on the path from the root of the decision process to terminal node $z \in Z$.

\runinsec{Strategies}
Conceptually, a strategy for an agent in a tree-form sequential decision process specifies a distribution $\vec{x}_j \in \Delta^{|A_j|}$ over the set of actions $A_j$ at each decision node $j \in \cJ$. We represent a strategy using the \emph{sequence-form representation}, that is, as a vector $\vec{q} \in \Rp^{|\Sigma|}$ whose entries are indexed by $\Sigma$. The entry $q[j,a]$ contains the product of the probabilities of all actions at all decision nodes on the path from the root of the process down to and including action $a$ at decision node $j\in \cJ$. A vector $\vec{q} \in\Rp^{|\Sigma|}$ is a valid sequence-form strategy if and only if it satisfies the constraints
(i) $\sum_{a\in A_j} q[j,a] = q[p_j]$ for all $j\in \cJ$; and (ii) $x[\emptyseq] = 1$~\citep{Romanovskii62:Reduction,Koller94:Fast,Stengel96:Efficient}.
We let $Q$ denote the set of valid sequence-form strategies.
Finally, we let $\Pi \subseteq Q$ denote the subset of sequence-form strategies whose entries are only $0$ or $1$; a strategy $\vec{\pi} \in \Pi$ is called a \emph{pure} sequence-form strategy, as it assigns probability $1$ to exactly one action at each decision node. 

    %\vspace{-2mm}
\section{Online Learning and Our Interactive Bandit Model}\label{sec:online learning}
%\vspace{-2mm}

%In this section, we review two popular online learning models, and introduce a third one which is natural in SDM problems. General introductions can be found~\cite{Hazan16:Introduction,Lattimore20:Bandit}.

In online learning, an agent interacts with its environment in this order:
    %\begin{enumerate}[nolistsep,itemsep=1mm,leftmargin=*]
    (i) The environment chooses a (secret) gradient vector $\vec{\ell}^t$ of bounded norm;
    (ii) The agent picks a pure strategy $\vec{\pi}^t \in \Pi$. The environment evaluates the reward (gain) of the agent as $(\vec{\ell}^t)^{\!\top}\! \vec{\pi}^t \in \bbR$;
     (iii) The agent observes some feedback about her reward. The feedback is used by the agent to learn to output good strategies over time.
    %\end{enumerate}
    %
The learning is online in the sense that the strategy $\vec{\pi}^t$ at time $t$ is output before any feedback for it (or future strategies) is available. A standard quality metric for evaluating an agent in this setting is the \emph{regret} that she accumulates over time:
%\begin{equation}\label{eq:def regret}
$
    R^T(\hat{\vec{\pi}}) \defeq \sum_{t=1}^T (\vec{\ell}^t)^{\!\top}\hat{\vec{\pi}} - \sum_{t=1}^T (\vec{\ell}^t)^{\!\top}\vec{\pi}^t$.
     %= \sum_{t=1}^T (\vec{\ell}^t)^{\!\top}(\hat{\vec{\pi}} - \vec{\pi}^t).
%\end{equation}
This measures the difference between the total reward accumulated up to time $T$, compared to the reward that would have been accumulated had the oracle output the fixed strategy $\hat{\vec{\pi}} \in \Pi$ at all times. A ``good'' agent, that is, a \emph{regret minimizer}, is one whose regret grows sublinearly: $R^T(\hat{\vec{\pi}}) = o(T)$ for all $\hat{\vec{\pi}} \in \Pi$.

%\runinsec{Applications of Low Regret}
%The online learning problem just described is a powerful and flexible mathematical abstraction. In addition to the natural application of using the regret minimizer to play an unknown opponent in an online, repeated fashion, the regret guarantee enjoyed by these oracles is known to be useful in offline problems too:
%\begin{itemize}[leftmargin=5mm,nolistsep,itemsep=1mm]
%  \item A Nash equilibrium in a two-player zero-sum extensive-form game can be computed by having the two players repeatedly play each other using the strategies recommended by two regret minimizers (one per player). Over time, the pair of average strategies output by the regret minimizers converges to a Nash equilibrium of the game. In fact, Nash equilibrium computation via regret minimization is currently the leading approach in large extensive-form games, and was a key component in several recent poker AI milestones~\citep{Bowling15:Heads,Moravvcik17:DeepStack,Brown17:Superhuman,Brown19:Superhuman}.
%  \item In a multiagent setting, when all players use regret minimization to play, over time the \emph{empirical frequency} of played strategies converges to a coarse-correlated equilibrium~      \citep{Moulin78:Strategically,Hart00:Simple}.
%  \item When a single agent uses a regret minimizer to play against a non-reactive opponent, the average played strategy converges to a best response to the opponent's underlying average strategy.
%\end{itemize}

%\subsection{Online Learning Models and Our Interactive Bandits Model}

Online learning models vary based on the type and extent of feedback that is made available to the agent.
%For the purposes of the present discussion, w DO YOU MEAN HERE OR THROUGHOUT THE PAPER ???
We will focus on two existing models---namely, the full-feedback\footnote{\label{fn:full information}In online learning, ``full-feedback'' is typically called ``full-information''. We use ``full-feedback'' to avoid confusion with full-information games, %(a.k.a. perfect-information games),
that is, games where the full state is available to all players at all times.} setting and the bandit linear optimization setting. Then we will introduce a third model that is especially natural for TFSDM.

\vspace{-3mm}
\paragraph{Full-Feedback\linkfootnote{fn:full information} Setting} Here, the environment always reveals the full gradient vector $\vec{\ell}^t$ to the agent (after the strategy has been output). This is the same setting that was proposed in the landmark paper by~\citet{Zinkevich03:Online} and is the most well-studied online optimization setting. Efficient agents that guarantee $O(T^{1/2})$ regret with high probability in the full-feedback setting are known well beyond TFSDM and extensive-form games (e.g, \citep{Shalev-Shwartz12:Online,Hazan16:Introduction}).
In fact, given any convex and compact set $\cX$, it is possible to construct an agent that outputs decisions $\vec{x}^t \in \cX$ that achieves $O(T^{1/2})$ regret \emph{in the worst case}, even when the gradient vectors $\vec{\ell}^t$ are chosen adversarially by the environment \emph{after} the decision $\vec{x}^t$ has been revealed.
In the specific context of TFSDM and extensive-form games, the most widely-used oracle in the full-feedback setting is based on the counterfactual regret (CFR) framework~\citep{Zinkevich07:Regret}. The idea is to decompose the task of computing a strategy for the whole decision process into smaller subproblems at each decision point. The local strategies are then computed via $|\cJ|$ independent full-feedback regret minimizers, one per decision node, that at each $t$ observe a specific feedback that guarantees low global regret across the decision process. CFR guarantees $O(T^{1/2})$ regret with high probability against any strategy $\vec{\pi} \in \Pi$.

%\todo{Talk about inversion of order 1,2???}

\vspace{-3mm}
\paragraph{Bandit Linear Optimization} Here, the only feedback that the environment reveals to the agent is the utility $(\vec{\ell}^t)^{\!\top} \vec{\pi}^t$ at each time $t$~\citep{Kleinberg04:Nearly,Flaxman05:Online}. Despite this extremely limited feedback, $\tilde{O}(T^{1/2})$ regret can still be guaranteed with high probability in some domains (including simplexes~\citep{Auer02:Nonstochastic} and spheres~\citep{Abernethy09:Beating}), although, to our knowledge, a polynomial algorithm that guarantees $\tilde{O}(T^{1/2})$ regret with high probability for any convex and compact domain has not been discovered yet.\footnote{Several algorithms are able to guarantee one or two out of the three requirements (i) applicable to any convex domain, (ii) polynomial time per iteration, (iii) $\tilde{O}(T^{1/2})$ regret with high probability. For example, \citet{Bartlett08:High} achieve (i) and (iii) by extending an earlier paper by~\citet{Dani08:Price}, and \cite{Gyorgy07:OnLine} achieves (ii) for the set of flows with suboptimal regret guarantees.} Guaranteeing $\tilde{O}(T^{1/2})$ \emph{in expectation} is possible for any domain of decisions~\citep{Abernethy08:Competing}, but unfortunately in-expectation low regret is not strong enough a guarantee to enable, for instance, convergence to Nash or correlated equilibrium as described above. In the specific case of tree-form sequential decision processes, \citep{Farina20:Counterfactual} proposed a bandit regret minimizer that achieves $O(T^{1/2})$ regret in expectation compared to any policy and linear-time iterations. Upgrading to in-high-probability $O(T^{1/2})$ regret guarantees while retaining linear-time iterations remains an open problem.

\vspace{-3mm}
\paragraph{Interactive Bandit} We propose \emph{interactive bandits} as a natural online learning model to capture the essence of tree-form sequential decision processes. Here, an agent interacts with the environment until a terminal state is reached, at which point the payoff (a real number) is revealed to the agent. The agent observes the environment's action (signal) whenever the interaction moves to an observation point. We formalize this as an online learning model as follows. At all times $t$, before the agent acts, the environment privately selects (i) a choice of payoff $u^t : Z \to \bbR$ for each terminal state $z \in Z$, and (ii) a secret choice of signals $s^t_k \in S_k$ for all observation points $k \in \cK$. These choices are hidden, and only the signals relevant to the observation points reached during the interaction will be revealed. Similarly, only the payoff $u^t(z^t)$ relative to the terminal state $z^t$ reached in the interaction will be revealed. In other words, the feedback that is revealed to the agent after the interaction is the terminal state $z^t$ that is reached (which directly captures all signals revealed by the environment, as they are the signals encountered on the path from the root of the decision process to $z^t$) and its corresponding payoff $u^t(z^t)$, which can be equivalently expressed as $u^t(z^t) = (\vec{\ell}^t)^{\!\top}\vec{\pi}^t$, where the gradient vector $\vec{\ell}^t$ is defined as the (unique) vector such that for all strategies $\vec{x} \in Q$,
\begin{equation}\label{eq:ib loss}
    (\vec{\ell}^t)^{\!\top}\vec{x} = \sum_{z\in Z} u^t(z) \mleft(\prod_{(k, s) \leadsto z} \bbone[s^t_k = s]\mright)x[\sigma(z)].
\end{equation}
We assume that the environment is adversarial, in the sense that the environment's choices of payoffs $u^t$ and signals $s^t_k$ at time $t$ can depend on the previous actions of the agent.

Our term ``interactive bandits'' comes from the fact that an algorithm for this setting can be thought of as interacting with the environment until a terminal state of the decision process is reached and a corresponding payoff is revealed. In other words, while for modeling purposes it is convenient to think about online learning algorithms as outputting strategies $\vec{\pi}$ for the \emph{whole} strategy space, one can think of an interactive bandits algorithm as one that instead only outputs one action at a time as the interaction moves throughout the decision process.

Since the interactive bandit model requires that the gradient vector $\vec{\ell}^t$ be expressible as in~\eqref{eq:ib loss}, it makes more assumptions on the gradient vector than either the bandit linear optimization model or the full-feedback model. In terms of feedback, it is an intermediate model: it receives a superset of the feedback that the bandit linear optimization framework receives, but significantly less than the full-feedback model. So, theoretically, one could always use an algorithm for the bandit linear optimization model to solve a problem in the interactive bandit model. However, as we show in this paper, one can design a regret-minimization algorithm for the interactive bandit model that achieves sublinear $O(T^{3/4})$ regret with high probability \emph{even in decision processes and extensive-form games whose decision space is at first unknown}. To our knowledge, no algorithm guaranteeing sublinear regret when the decision space is at first unknown has been designed for bandit linear optimization.

    \section{Algorithm for Unknown Tree-Form Sequential Decision Making Problems}\label{sec:algorithm}

We now describe a regret minimizer for the interactive bandit model.
%introduced in \cref{sec:online learning}.
At all time $t$, it goes through two phases: first, the \emph{rollout} phase, and then the \emph{regret update} phase. During the rollout phase, the algorithm plays until a terminal state $z^t$ and its corresponding payoff $u^t$ is revealed. During the regret update phase, the algorithm rewinds through the decision nodes $j$ encountered during that trajectory, and updates certain parameters at those decision nodes based on the newly-observed payoff $u^t$.%
\footnote{While our algorithm shares many of the building blocks on Monte Carlo Tree Search (MCTS)---incremental tree-building, back-propagation, rollouts---it is \emph{not} an anytime search algorithm, at least not in the sense of traditional game-tree search like the one employed by the Online Outcome Sampling algorithm~\citep{Lisy15:Online}.}
Like the CFR framework, our algorithm picks actions at each decision node $j$ by means of a \emph{local} full-feedback regret minimizer $\cR_j$ for the action set $A_j$ at that decision node.

    \subsection{Rollout Phase: Playing the Game}
    We describe two alternative algorithms for the rollout phase (namely the ``upfront-flipping'' and the ``on-path-flipping'' rollout variants), which differ in the way they mix exploration and exploitation. Both variants are theoretically sound, and yield to the same sublinear in-high-probability regret bound (\cref{sec:guarantees}), even when different variants are used at different times $t$ while playing against the adversarial opponent.
    
    We start from the ``upfront-flipping'' variant, which is arguably the conceptually simpler variant, although we find it to usually perform worse in practice.
    
    \vspace{-3mm}
    \paragraph{Upfront-Flipping Rollout Variant}
    When the \emph{upfront-flipping rollout variant} is used at time $t$, at the beginning of the rollout phase and before any action is played, a biased coin is tossed to decide the algorithm used to play out the interaction:
    \begin{itemize}[nolistsep,itemsep=1mm,leftmargin=5mm]
      \item With probability $\beta^t$, the \textsc{Explore} routine is used. It ignores the recommendations of the local regret minimizers at each decision node and instead plays according to the \emph{exploration function} $h^t: \Sigma \to \bbR_{>0}$ (more details are below). In particular, at every decision node $j$ encountered during the rollout, the agent picks action $a \in A_j$ at random according to the distribution $h^t(j,a) / (\sum_{a'\in A_j} h^t(j,a'))$.
          %Again, whenever the decision maker reaches a decision node $j$ that was not entered before, a new full-feedback regret minimizer $\cR_j$ for the set of action $A_j$ is created for that decision node.
      \item With probability $1-\beta^t$, the \textsc{Exploit} routine is used. With this routine, at every decision node $j$ encountered during the rollout, the agent picks a decision by sampling from the distribution $\vec{x}_j^t \in \Delta^{|A_j|}$ recommended by the regret minimizer $\cR_j$.
          %Whenever the agent reaches a decision node $j$ that was not entered before, a new regret minimizer $\cR_j$ is created for that decision node.
    \end{itemize}
    In both cases, a regret minimizer $\cR_j$ for decision node $j$ is created when $j$ is first discovered.

    Our upfront-flipping rollout strategy differs from the $\epsilon$-greedy strategy in that the coin is tossed for the entire trajectory, not at each decision point.
     
     \vspace{-3mm}
    \paragraph{On-Path-Flipping Rollout Variant}
    In the \emph{on-path-flipping} rollout variant, there is no single coin toss to distinguish between exploration and exploitation, and the two are interleaved throughout the rollout.
    Before any action is picked, the two \emph{reach} quantities $r^t, \hat{r}^t$ are both set to $1$. Then, the rollout phase begins, and eventually the agent will be required to make a decision (pick an action) at some decision point $j$. Let $\vec{x}_j^t \in \Delta^{|A_j|}$ be the distribution over actions recommended by the regret minimizer $\cR_j$. In the on-path-flipping rollout variant, the agent picks an action $a \in A_j$ at $j$ with probability proportional to
    \[
        (1 - \beta^t) r^t \cdot {x}_j^t[a] + \beta^t \hat{r}^t\cdot \frac{h^t(j,a)}{\sum_{a'\in A_j} h^t(j,a')}.
    \]
    Let $a^*$ be the chosen action; $r^t$ and $\hat{r}^t$ are updated according to the formulas
    $
        r^t \gets r^t\cdot x_j^t[a^*]$ and $\hat{r}^t \gets \hat{r}^t \cdot \frac{h^t(j,a^*)}{\sum_{a'\in A_j} h(j,a')}$.
    The agent keeps using this specific way of selecting actions and updating the reach quantities $r,\hat{r}$ for all decision points encountered during the rollout.

    \vspace{-3mm}
    \paragraph{The role of $h^t$}
    In both variants, the role of the exploration function $h^t$ is to guide exploration of different parts of the decision process.\footnote{Despite the positive exploration term induced by $h^t$, it is not guaranteed that all decision points will be discovered as $T \to \infty$ as the adversarial environment might prevent so. Nonetheless, our algorithm guarantees sublinear regret with high probability.} The optimal choice for $h^t$ is to have $h^t(j,a)$ measure the number of terminal states in the subtree rooted at $(j,a)$. When this information is not available, a heuristic can be used instead. If no sensible heuristic can be devised, the uniform exploration strategy $h^t(j,a) = 1$ for all $(j,a)\in\Sigma$ is always a valid fallback. In \cref{thm:regret bound general} below, we give guarantees about the regret cumulated by our algorithm that apply to any $h^t:\Sigma\to\bbR_{>0}$.

     \subsection{Regret Update Phase: Propagating the Payoff up the Tree}
     In the regret update phase, the revealed feedback (that is, the revealed utility $u^t(z^t)$ and the terminal state $z^t$ that was reached in the rollout phase) is used to construct suitable \emph{local} gradient vectors $\vec{\ell}_j^t$ for each of the local regret minimizers $\cR_j$ on the path from the terminal state $z^t$ up to the root. Let $(j_1, a_1), \dots, (j_m, a_m)$ be the sequence of decision nodes and actions that were played, in order, that ultimately led to terminal state $z^t$ during the repetition of the game at time $t$. We start to construct local gradient vectors from decision node $j_m$, where we set the three quantities
\begin{align*}
    &\gamma^t \defeq (1-\beta^t)\cdot\prod_{i=1}^m x_{j_i}^t[a_i] + \beta^t\cdot\prod_{i=1}^m \frac{h^t(j_i,a_i)}{\sum_{a' \in A_{j_i}} h^t(j_i,a')},\\
    &\hspace{2cm}\hat{u}_{j_m}^t \defeq \frac{u^t(z^t)}{\gamma^t},\quad
    \vec{\ell}_{j_m}^t \defeq \hat{u}^t_{j_m}\vec{e}_{a_{m}},
\end{align*}
where we used the notation $\vec{e}_{a_m} \in \Delta^{A_j}$ to denote the $a_m$-th canonical vector, that is the vector whose components are all $0$ except for the $a_m$-th entry, which is $1$. Then, for all $i = 1, \dots, m-1$, we recursively let
\[
    \hat{u}^t_{j_i} \defeq x^t_{j_i}[a_{i}]\cdot \hat{u}^t_{j_{i+1}},\qquad
    \vec{\ell}_{j_i}^t \defeq \hat{u}^t_{j_i}\vec{e}_{a_{i}}.
\]
Finally, for all $i = 1,\dots, m$, the gradient vector $\vec{\ell}_{j_i}^t$ is revealed as feedback to the local full-feedback regret minimizer $\cR_{j_i}$ at decision point $j_i$.

% \begin{observation}\label{obs:xxx1}
%     For each decision point $j$ touched during the rollout phase (for either rollout variant) and in the regret update phase, the agent performs an amount of operations linear in the number of actions $|A_j|$ at $j$. So, the rollout phase and the regret update phase require work at most linear in the size of the underlying SDM.
% \end{observation}

    \subsection{Average Policy}
    When regret minimizers are used to solve a convex-concave saddle point problem (such as a Nash equilibrium in a two-player zero-sum game), only the profile of \emph{average} policies produced by the regret minimizers are guaranteed to converge to the saddle point. For this reason, it is crucial to be able to represent the \emph{average} policy of an agent. Since we are assuming that the structure of the decision problem is only partially known, this operation requires more care in our setting. As we now show, it is possible to modify the algorithm so that the average policy can be extracted.

    In order to maintain the average policy, we maintain an additional vector $\bar{\vec{x}}_j$ at each discovered decision node $j$. Intuitively, these vector will be populated with entries from the cumulative sum of all partial sequence-form strategies recommended so far by the $\cR_j$'s.
    As soon as $j$ is discovered for the first time (say, time $t$), we create its $\bar{\vec{x}}_j$.
    If $j$'s parent sequence is the empty sequence (that is, $j$ is one of the possible first decision nodes in the TFSDM process, i.e., $j$ is preceded only by observation nodes), we simply set $\bar{x}_j[a]^t \defeq t / |A_j|$ for all $a \in A_j$. Otherwise, let $p_j = (j',a')$ be the parent sequence of $j$, and we set $\bar{x}_j[a]^t \defeq \bar{x}_{j'}[a'] / |A_j|$ for all $a \in A_j$. Then, at all times $t$, after the feedback has been received but before the regret update phase has started, we introduce a new \emph{average policy update phase}. In it, we iterate through all the  decision nodes $j_i$ that have been discovered so far (that is, the union of all decision points discovered up to time $t$), in the order they have been discovered. For each of them, we update  $\bar{\vec{x}}_{j_i}$ according to the following rule. Let $\vec{x}_{j_i}$ be the policy strategy returned by the local full-feedback regret minimizer $\cR_{j_i}$. If $j_i$'s parent sequence is the empty sequence, we simply set $\bar{\vec{x}}_{j_i}^{t+1} \defeq \bar{\vec{x}}_{j_i}^t + \vec{x}_{j_i}$ and $r^t_{j_i} \defeq \vec{x}_{j_i}$.
Otherwise, let $p_{j_i} = (j',a')$ be the parent sequence of $j$, and we set $\bar{\vec{x}}_{j_i}^{t+1} \defeq \bar{\vec{x}}_{j_i}^t + r^{t}_{j'}[a'] \cdot \vec{x}^t_{j_i}$, and $r^t_{j_i}[a] \defeq r^t_{j'}[a'] \cdot x^t_{j_i}[a]$ for all $a \in A_{j_i}$.

In order to play the average policy, it is enough to play actions proportionally to $\bar{\vec{x}}_j^{t+1}$ at all discovered decision nodes $j$, and actions picked uniformly at random at undiscovered decision nodes. 

\begin{observation}\label{obs:xxx2}
    In all phases, the agent performs an amount of operations at most linear in the number of actions $|A_j|$ at each decision point $j$ discovered up to time $t$. So, the average policy update phase requires work at most linear in the size of the underlying TFSDM.
\end{observation}

% Combining Observations \ref{obs:xxx1} and \ref{obs:xxx2} we conclude that the algorithm we described runs in linear time per iteration $t$.

    \section{Guarantees on Regret}\label{sec:guarantees}
        
%We now present regret bounds that hold at all times $T$ for our algorithm from \cref{sec:algorithm}.
%The bounds do not depend on which rollout variant is chosen at each time $t$, and different variants can be chosen freely at different times.
%
In this section, we provide a sketch of the key logical steps in our proof of the sublinear regret guarantees for the algorithm we just described in \cref{sec:algorithm}.\footnote{More details and full proofs are available in the full version of this paper, which
is available on arXiv.}
In order to make the analysis approachable, we start by presenting a \emph{conceptual} version of our algorithm under the assumption that the structure
of the tree-form sequential decision problem is fully known. As it turns out, our conceptual algorithm can be implemented exactly as in \cref{sec:algorithm}, and therefore, it does not actually need to know the structure of the decision process in advance.

At a high level, the construction of our conceptual interactive bandit regret minimizer works as follows (see also \cref{fig:analysis}).  At each iteration $t$, we use a full-feedback regret minimizer $\cR_Q$ to output a recommendation for the next sequence-form strategy $\vec{y}^t \in Q$ to be played. Then, we introduce a bias on $\vec{y}^t$, which can be thought of as an exploration term. The resulting biased strategy is $\vec{w}^t \in Q$. Finally, we sample a deterministic policy $\vec{\pi}^t \in \Pi$ starting from $\vec{w}^t$. After the playthrough is over and a terminal node $z^t \in Z$ has been reached, we use the feedback (that is, the terminal node $z^t$ reached in the playthrough, together with its utility $u^t(z^t)$), to construct an unbiased estimator $\tilde{\vec{\ell}}^t$ of the underlying gradient vector that was chosen by the environment.

    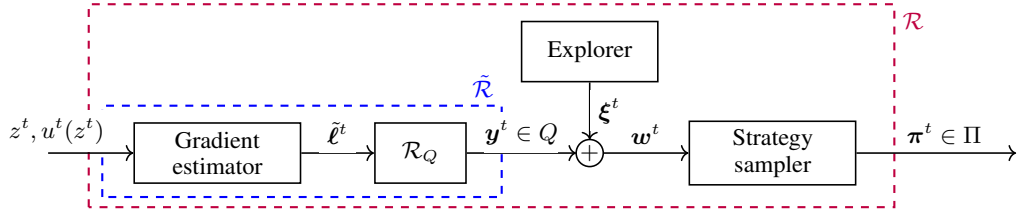
\begin{figure*}[t]\centering
        \begin{tikzpicture}[scale=.9]\small
          \node[draw,semithick,align=center,text width=2.0cm,minimum height=.8cm] (G) at (0,0) {Gradient estimator};
          \node[draw,semithick,align=center,text width=1.0cm,minimum height=.8cm] (Q) at (3,0) {$\cR_Q$};
          \node[draw,semithick,align=center,text width=2.0cm,minimum height=.8cm] (S) at (8.2,0) {Strategy sampler};
          \node[draw,semithick,align=center,text width=1.6cm,minimum height=.8cm] (E) at (5.5,1.5) {Explorer};
          \node[] (plus) at (5.5,0) {$+$};
          \draw[thick,dashed,blue] (-1.7,-.65) rectangle (4.2,.7);
          \node[above left,blue] at (4.2,.7) {$\tilde{\cR}$};
          \draw[thick,dashed,purple] (-1.9,-.8) rectangle (10,2.2);
          \node[below right,purple] at (10,2.2) {$\cR$};

          \draw[semithick] (5.5, 0) circle (.2);
          \draw[semithick,->] (E.south) --node[right,fill=white]{$\vec{\xi}^t$} (5.5,.2);
          % \draw[semithick,->] (-2.5,0) --node[pos=.1,above,fill=white]{$(\vec{\ell}^t)^\top \vec{\pi}^t$} (G.west);
          \draw[semithick,->] (-2.5,0) --node[pos=.1,above,fill=white]{$z^t, u^t(z^t)$} (G.west);
          \draw[semithick,->] (G.east) --node[above,fill=white]{$\tilde{\vec{\ell}}^t$} (Q.west);
          \draw[semithick,->] (Q.east) --node[above,fill=white,inner xsep=0]{$\vec{y}^t\in Q$} (5.3,0);
          \draw[semithick,->] (5.7,0) --node[above,fill=white]{$\vec{w}^t$} (S.west);
          \draw[semithick,->] (S.east) --node[above,fill=white,pos=.55,inner xsep=0]{$\vec{\pi}^t\in\Pi$} (11.8,0);
        \end{tikzpicture}
        \caption{Conceptual construction of our algorithm for the interactive bandit online learning setting.}\label{fig:analysis}
    \end{figure*}

    Our full-feedback regret minimizer $\cR_Q$ is the counterfactual regret minimization (CFR) algorithm~\citep{Zinkevich07:Regret}. At all times $t$, CFR
    combines the strategies output by the local regret minimizers $\cR_j$ ($j\in \cJ$), one per decision node, into the strategy $\vec{y}^t\in Q$ for the overall TFSDM problem. CFR guarantees $O(T^{1/2})$ regret in the worst case against any strategy $\hat{\vec{y}} \in Q$. %CFR has the advantage that no strategy updates from the local regret minimizers $R_j$ are required when  $j$ has not been encountered.
    
The gradient estimator that we use in our algorithm is a form of importance-sampling estimator that can be constructed starting from the feedback received by the regret minimizer (that is, the terminal leaf $z^t$ and its corresponding payoff $u^t(z^t)$). It is a particular instantiation of the \emph{outcome sampling} estimator that appeared in the works by~\citet{Lanctot09:Monte} and~\citet{Farina20:Stochastic}, and is defined as
    \begin{equation}\label{eq:estimator}
        \tilde{\vec{\ell}}^t \defeq \frac{u^t(z^t)}{w^t[\sigma(z^t)]}\vec{e}_{\sigma(z^t)},
    \end{equation}
    where $\vec{e}_{\sigma(z^t)} \in \bbR^{|\Sigma|}$ is the vector that has zeros
    everywhere but in the component corresponding to the terminal sequence $\sigma(z^t)$, where
    it has a one.

The role of the exploration term is to reduce the norm of the gradient estimators, by
making sure that the denominator $w^t[\sigma(z^t)]$ in~\eqref{eq:estimator} is never too small. We bias the strategies $\vec{y}^t$ output by $\cR_Q$ by taking a convex combination
\[
    \vec{w}^t \defeq (1-\beta^t)\cdot \vec{y}^t + \beta^t \cdot \vec{\xi}^t
\]
with the \emph{exploration strategy} $\vec{\xi}^t \in Q$ defined as the sequence-form strategy that picks action $a \in A_j$ at decision node $j\in\cJ$ with probability proportional to
\[
    \xi^t[j,a] \propto \frac{h^t(j,a)}{\sum_{a\in A_j} h^t(j,a)}.
\]

The role of the strategy sampler is to provide an \emph{unbiased} estimator $\vec{\pi}^t \in \Pi$ of $\vec{w}^t$. In fact, any unbiased estimator can be used in our framework, and in \cref{sec:algorithm} we gave two: upfront-flipping and on-path-flipping strategy sampling.

    Our analysis is split into two main conceptual steps. First, we quantify the regret degradation that is incurred in passing from $\tilde{\cR}$ to $\cR$ (\cref{fig:analysis}). Specifically, we study how the regret cumulated by $\tilde{\cR}$,
    $
        \tilde{R}^T(\vec{\pi}) \defeq \sum_{t=1}^T (\vec{\ell}^t)^\top(\vec{\pi} - \vec{y}^t),
    $
    is related to the regret cumulated by our overall algorithm,
    $
        R^T(\vec{\pi}) \defeq \sum_{t=1}^T (\vec{\ell}^t)^\top (\vec{\pi} - \vec{\pi}^t).
    $
    In this case, the degradation term comes from the fact that the strategy output by $\cR$, that is, $\vec{\pi}^t \in \Pi$, is not the same as the one, $\vec{y}^t$, that was recommended by $\tilde{\cR}$, because an exploration term was added. Using a concentration argument, the following inequality can be shown.
    \begin{proposition}\label{prop:relationship R Rtilde}
        At all $T$, for all $p\in(0,1)$ and $\vec{\pi} \in \Pi$, with probability at least $1-p$,
        \[
            R^T(\vec{\pi}) \le \tilde{R}^T(\vec{\pi}) +  \Delta\mleft( \sqrt{2T\log\frac{1}{p}}+ \sum_{t=1}^T \beta^t\mright).
        \]
    \end{proposition}

The second step in the analysis is the quantification of the regret degradation that is incurred in passing from $\cR_Q$ to $\tilde{\cR}$ (\cref{fig:analysis}). Specifically, we will study how the regret cumulated by $\tilde{\cR}$ is related to the regret
    $
        \cR_Q^T(\vec{\pi}) \defeq \sum_{t=1}^T (\tilde{\vec{\ell}}^t)^\top(\vec{\pi} - \vec{y}^t),
    $
    which is known to be $O(T^{1/2})$ with high probability from the analysis of the  CFR algorithm. In particular, the following can be shown with a second concentration argument.
    \begin{proposition}\label{prop:relationship Rtilde RQ}
        At all $T$, for all $p\in(0,1)$ and $\vec{\pi} \in \Pi$, with probability at least $1-p$,
        \[
            \tilde{R}^T(\vec{\pi}) \le R_Q^T(\vec{\pi}) + \frac{\Delta}{\beta^T}\nu \,\sqrt{2T\log\frac{1}{p}},
        \]
        where
        \[
            \nu \defeq \sqrt{\frac{1}{T}\sum_{t=1}^T \max_{z\in Z} \mleft\{ \prod_{(j,a) \leadsto z} \mleft(\frac{\sum_{a'\in A_{j}} h^t(j, a')}{h^t(j,a)}\mright)^{\!\!2}\mright\}}.
        \]
    \end{proposition}
When $h^t$ measures the number of terminal states reachable under any sequence $(j,a)$, the constant $\nu$ satisfies $\nu \le |\Sigma| - 1$. For the uniform exploration function ($h^t(j,a) = 1$ for all $j,a$), $\nu$ is upper bounded by the product of the number of actions at all decision nodes, a polynomial quantity in ``bushy'' decision processes.

Combining \cref{prop:relationship R Rtilde} and \cref{prop:relationship Rtilde RQ} using the union bound lemma, we obtain a general regret bound for $\mathcal{R}$ that holds for any choice of local regret minimizers $\cR_j$ ($j \in \cJ$), non-increasing stepsizes $\beta^t$, and explorations functions $h^t$ at all times $t$, which we formalize in \cref{thm:regret bound general}.  It will be the basis for \cref{thm:regret specific}, which provides a way to set the algorithm parameters to achieve sublinear $O(T^{3/4})$ regret.

        \begin{theorem}\label{thm:regret bound general}
            Let $R_j^T(\hat{\vec{\pi}}_j)$ denote the regret cumulated by the local full-feedback regret minimizer $\cR_j$ compared to a generic strategy $\hat{\vec{\pi}}_j$, and $\Delta$ be the maximum range of payoffs that can be selected by the environment at all times. Then, for all $T \ge 1$, $p \in (0,1)$, and $\hat{\vec{\pi}}\in Q$, with probability at least $1-p$ the regret cumulated by the algorithm of \cref{sec:algorithm} satisfies
                  \begin{align*}\label{eq:bound general}
                    %\boxed{
                        &R^T\!(\pihat) \le \max_{\vec{q} \in
                        \Pi}\mleft\{\sum_{j \in \cJ} q[p_j]\cdot \max_{\hat{\vec{\pi}}_j\in \Delta^{|A_j|}} R_j^T(\pihat_j)\mright\} \\[-1mm]
                        &\hspace{2.5cm}+ \frac{\Delta}{\beta^T}\mleft(1 + \nu\mright)\sqrt{2T \log\frac{2}{p}} + \Delta\sum_{t=1}^T \beta^t
                    ,
                  \end{align*}
                  where $\nu$ is as in \cref{prop:relationship Rtilde RQ}.
        \end{theorem}
       \noindent In \cref{thm:regret specific} we operationalize \cref{thm:regret bound general} by showing sensible choices of stepsizes $\beta^t$ and local regret minimizers $\cR_j$.

        \begin{restatable}{theorem}{thmregretspecific}\label{thm:regret specific}
            Let the local full-feedback regret minimizers $\mathcal{R}_j$ guarantee $O(T^{1/2})$ regret in the worst case\footnote{Valid choices include the following algorithms: regret matching~\citep{Hart00:Simple}, regret matching$^+$~\citep{Tammelin15:Solving}, follow-the-regularized-leader, online mirror descent, exponential weights, hedge, and others.}, and let $p\in (0,1)$. Furthermore, let the exploration probabilities be
$\beta^t \defeq \min\{1, k\cdot t^{-1/4}\}$ for all $t$, where $k > 0$ is an arbitrary constant.
Then, there exists a (decision-problem-dependent) constant $c$ independent of $p$ and $T$ such that for all $T \ge 1$
            \[
                \bbP\mleft[\max_{\hat{\vec{\pi}} \in \Pi} R^T(\hat{\vec{\pi}}) \le c \cdot T^{3/4}\Delta\sqrt{\log\frac{2}{p}}\mright] \ge 1 - p.
            \]
            When $h^t$ is an exact measure of the number of terminal states, $c$ is polynomial in $|\Sigma|$. Otherwise, it is linear in the constant $\nu$ defined in~\cref{prop:relationship Rtilde RQ}, which depends on the specific exploration functions used.
        \end{restatable}

Since \cref{thm:regret specific} guarantees sublinear regret with high probability, our algorithm can be used for all purposes described in \cref{sec:online learning}, including computing an approximate Nash equilibrium in a two-player zero-sum extensive-form game whose structure is \emph{a priori} unknown.

    \section{Empirical Evaluation}

In our experiments, we used our algorithm to compute an approximate Nash equilibrium. We compared our method to established model-free algorithms from the multiagent reinforcement learning and computational game theory literature for this setting: \emph{neural fictitious self-play (NFSP)}~\cite{Heinrich16:Deep}, the \emph{policy gradient (PG)} approach of \citet{Srinivasan18:Actor}, and the \emph{online variant of Monte-Carlo CFR (online MCCFR)} mentioned in~\citep{Lanctot09:Monte}. In line with prior empirical evaluations of those methods, we compare the algorithms on two standard benchmark games: Kuhn poker~\citep{Kuhn50:Simplified} and Leduc poker~\citep{Southey05:Bayes}. The games are reviewed in \cref{app:games} in the full version of this paper.

We used the implementations of PG and NFSP provided in \emph{OpenSpiel}~\citep{Lanctot19:OpenSpiel}.%, a framework for multiagent reinforcement learning.
They internally use Tensorflow. For PG, we tested the RPG and QPG policy gradient formulations, but not the RM formulation (it performed worst in the original paper~\cite{Srinivasan18:Actor}).  We implemented online MCCFR and our algorithm in C++ (online MCCFR is not implemented in OpenSpiel). We ran every algorithm with five random seeds. In the figures below, we plot the average exploitability (a standard measure of closeness to equilibrium) of the players averaged across the five seeds. The shaded areas indicate the \textit{maximum} and \textit{minimum} over the five random seeds.
%
%\textbf{Hyperparameters}\quad
For NFSP we used the hyperparameter recommended by the OpenSpiel implementation.
For Kuhn poker, we used the settings for PG that were tuned and found to work the best by~\citet{Srinivasan18:Actor}---they are publicly available through OpenSpiel. For PG in Leduc poker, we performed a hyperparameter sweep and selected for the two PG plot (RPG and QPG formulation) the best combination hyperparameters (full details are in the appendix).
For both online MCCFR and our algorithm, we used RM+~\citep{Tammelin14:Solving} as the local (full-feedback) regret minimizer. For our algorithm, we only show performance for the on-path-flipping variant. The upfront-flipping variant performed significantly worse and data is available in the appendix of the full version of this paper. We tested $k \in \{0.5,1,10,20\}$ and set $h^t$ to either the uniform exploration function ($h^t$ constant) or the theoretically-optimal exploration function $h^t$ that measures the number of terminal nodes as explained in \cref{sec:algorithm}. The performance of the two exploration functions was nearly identical, so in \cref{fig:experiments} we show our algorithm with the uniform exploration function. We chose $k=10$ since that performed well on both games.  The plots for all other hyperparameter combinations for our algorithm are in the appendix. For online MCCFR, the only hyperparameter is $\epsilon$, which controls the $\epsilon$-greediness of the exploration term added before sampling and outputting the strategy at each time $t$. We tested $\epsilon = 0.6$ (which was found useful for the different \emph{self-play} MCCFR algorithm~\citep{Lanctot09:Monte}), $0.1$, and $0.0$ (which corresponds to pure exploitation); \cref{fig:experiments} shows the setting that performed best.

\begin{figure}[t]
    % \vspace{-1mm}
  \centering
  \includegraphics[scale=.67]{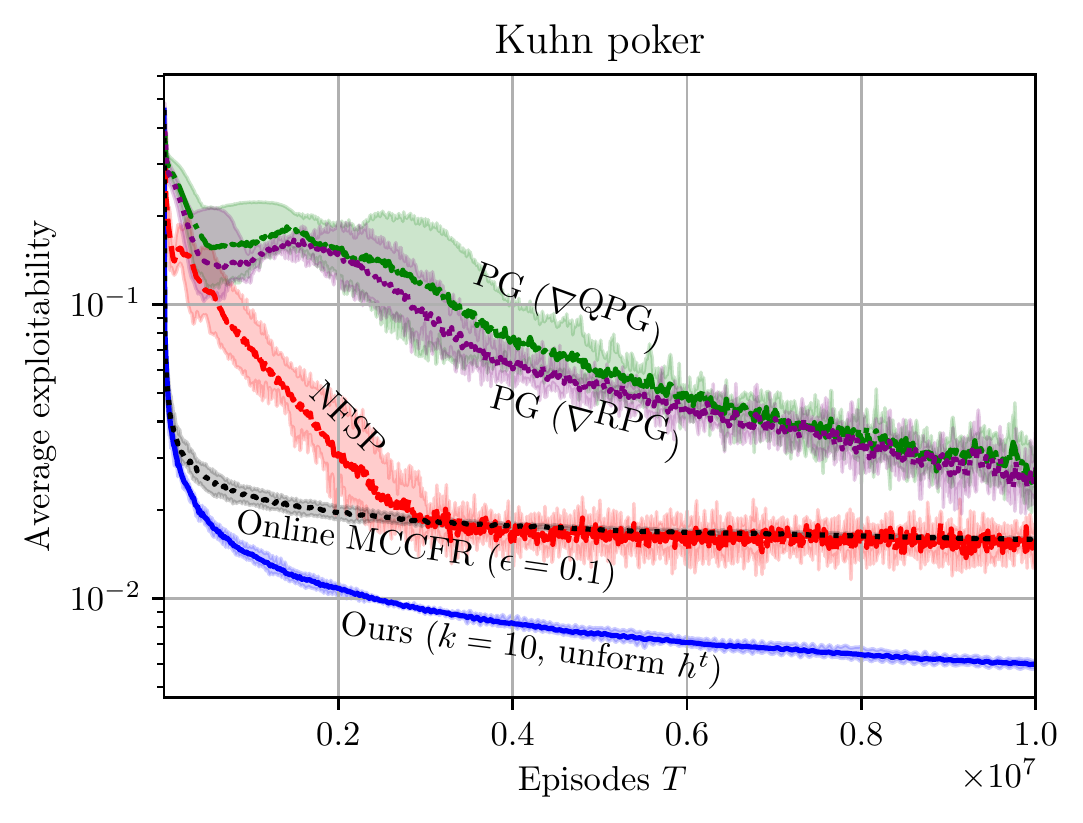}\\[-1mm]
  \includegraphics[scale=.67]{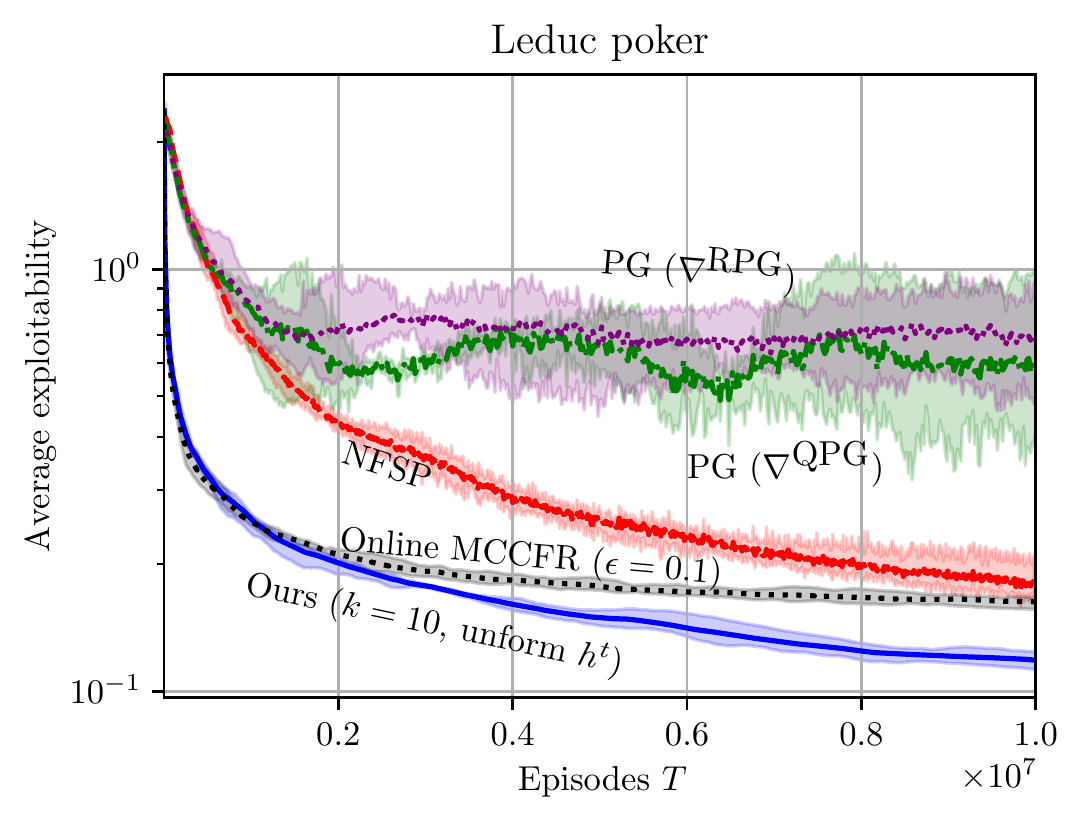}
  \vspace{-3mm}
  \caption{Comparison of the algorithms.}\label{fig:experiments}
% \vspace{-3mm}
\end{figure}

%\textbf{Discussion}\quad
Out of all the algorithms, ours is the only one that guarantees sublinear regret with high probability. This superior guarantee appears to translate into superior practical performance as well. In both benchmark games, our algorithm has lowest exploitability, often by a factor 2x-4x.

%In Kuhn poker, NFSP dominates both policy gradient formulations, but not MCCFR with $0.1$-greedy exploration. In Leduc poker, the policy gradient methods are initially superior to NFSP, but eventually NFSP overtakes.

%Finally, in our experiments we found the runtime of each iteration of both our algorithm and MCCFR to be roughly two orders of magnitude faster than the other methods, though it is not clear what fraction of that overhead could be overcome by using dedicated hardware (such as a GPU) for training. 

    % \vspace{-1mm}
\section{Conclusions and Future Research}
% \vspace{-0mm}

We introduced a new online learning model, which we coined the \emph{interactive bandit model}, that captures tree-form sequential decision making.
We developed an algorithm that guarantees sublinear $O(T^{3/4})$ regret with high probability in this model, even when the structure of the underlying decision problem is at first unknown to the agent and must be explored as part of the learning process. This is, to our knowledge, the first in-high-probability regret minimizer for this setting.
It can be used for multiagent reinforcement learning. Its regret guarantee enables it to be used in any application for which regret minimization is useful: approximating Nash equilibrium or quantal response equilibrium~\cite{Ling18:What,Farina19:Online} in two-player zero-sum games, approximating coarse correlated equilibrium in multi-player games~\cite{Moulin78:Strategically,Celli19:Learning}, learning a best response, safe opponent exploitation~\cite{Farina19:Online}, online play against an unknown opponent/environment, etc.
It is open whether better than $O(T^{3/4})$ regret is achievable in this important setting.

    \section*{Acknowledgments}
    This material is based on work supported by the National Science Foundation
    under grants IIS-1718457, IIS-1901403, and CCF-1733556, and the ARO under
    award W911NF2010081. Gabriele Farina is supported by a Facebook fellowship.

    We are grateful to Marc Lanctot and Marcello Restelli for their valuable feedback while preparing our
    manuscript, and to Marc Lanctot and 
    Vinicius Zambaldi for their help in tuning the hyperparameters and running experiments
    for the policy gradient algorithm of \citet{Srinivasan18:Actor}.

    \bibliography{dairefs}

\iftrue
    \clearpage
    \onecolumn
    \leftlinenumbers
    \appendix

    \addtolength{\hoffset}{1.5cm}
    \addtolength{\textwidth}{-3cm}
    \addtolength{\hsize}{-3cm}
    \addtolength{\linewidth}{-3cm}
    \addtolength{\columnwidth}{-3cm}

    \section{Summary of Online Learning Models}\label{app:online learning models}
    \begin{table}[H]\centering
      \scalebox{.90}{\begin{tabular}{p{2.9cm}p{6.6cm}p{5.2cm}}
        \toprule
            \raggedright\bfseries Online learning model & \raggedright\bfseries Assumptions on gradient vector $\vec{\ell}^t$ \textcolor{gray}{\normalfont
(other than bounded norm)}& \bfseries Revealed feedback\\
        \midrule
                Full-feedback\linkfootnote{fn:full information}
            &
                unconstrained: $\vec{\ell}^t \in \bbR^{|\Sigma|}$\newline
            &
                $\vec{\ell}^t$
            \\\colorrule{black!60!white}
                Bandit linear\newline optimization
            &
                unconstrained: $\vec{\ell}^t \in \bbR^{|\Sigma|}$\newline
            &
                $(\vec{\ell}^t)^{\!\top}\vec{\pi}^t \in \bbR$
            \\\colorrule{black!60!white}
                Interactive bandit\newline\textcolor{gray}{(this paper)} &
                $u^t\!: Z \to \bbR$ choice of payoff at each $z\!\in\! Z$\newline
                $s^t_k \in S_k$\hspace{4.5mm} choice of signal at each $k\!\in\! \cK$\vspace{2mm}\newline
                $\vec{\ell}^t$ as in~\cref{eq:ib loss2}
            & \begin{minipage}[t]{\columnwidth}
              \begin{itemize}[topsep=0pt,leftmargin=*]
                \item $z^t \in Z$ terminal state,
                \item $u^t(z^t) = (\vec{\ell}^t)^{\!\top\!} \vec{\pi}^t \in \bbR$ payoff at $z^t$
              \end{itemize}
              \end{minipage}
            \\\bottomrule
      \end{tabular}}

      \vspace{1mm}
      \caption{Comparison between different online learning models and their assumptions on the gradient vector and revealed feedback.}
      \label{tab:comparison of models}
    \end{table}

    \begin{observation} 
    In online learning, the environment picks a secret loss function (gradient vector) $\vec{\ell}^t$ at every iteration $t$. The gradient vector controls the utility of the player, which is defined as the inner product $(\vec{\ell}^t)^\top \vec{\pi}^t$, where $\vec{\pi}^t$ is the strategy output by the agent at iteration $t$.
In the interactive bandit model, it is perhaps more natural to think of the environment as picking a choice of signal $s_k^t$ at each observation node $k\in\cK$, as well as a choice of payoff $u^t(z)$ for each terminal node $z \in Z$. The utility of the agent is then computed as $u^t(z^t)$, where $z^t$ is the (unique) terminal node that is reached when the agent plays according to $\vec{\pi}^t \in \Pi$, and the environment sends the signals $s_k^t$.

The two points of view can be reconciled. In particular, at every iteration $t$ it is possible to identify a gradient vector $\vec{\ell}^t$ for the interactive bandit setting such that $u^t(z^t) = (\vec{\ell}^t)^\top \vec{\pi}^t$, as follows.
    As mentioned in the body, let $\sigma(z)$ be the last sequence (decision node-action pair) on the path from the root of the decision process to terminal node $z \in Z$, and consider the gradient vector $\vec{\ell}^t$, defined as the (unique) vector such that
\begin{equation}\label{eq:ib loss2}
    (\vec{\ell}^t)^{\!\top}\vec{x} = \sum_{z\in Z} u^t(z) \mleft(\prod_{(k, s) \leadsto z} \bbone[s^t_k = s]\mright)x[\sigma(z)]\qquad\forall \vec{x} \in Q.%\footnote{The convex hull $\co\Pi$ of the set of pure strategies is known to be equal to the set of sequence-form strategies $Q$.}
\end{equation}
    For each terminal node $z \in Z$ in the argument of the sum, the product in parenthesis is $1$ if and only if signals $s_k^t$ picked by the environment are a superset of those on the path from the root to $z$. Furthermore, when $\vec{x} = \vec{\pi}^t$, $x[\sigma(z^t)] = 1$ if and only if the decision maker has decided to play all actions on the path from the root to $z^t$. So, the only terminal node for which the right hand side of \cref{eq:ib loss2} is nonzero is $z = z^t$. Consequently, $(\vec{\ell}^t)^\top \vec{\pi}^t = u^t(z^t)$. 
\end{observation}

        \section{Analysis and Proofs}

    In order to make the proof approachable, at first we will assume that the
    structure of the decision process is fully known. We will then present and analyze an algorithm that enjoys the properties described in \cref{thm:regret bound general}. Finally, we will show that the algorithm can be implemented as described in \cref{sec:algorithm}, and therefore, that it does not actually need to know the structure of the decision process.

    % While in the body we only had enough space to show the end result, our algorithm for the interactive bandit model is actually made of several independent conceptual parts, which are summarized in \cref{fig:analysis}. In the camera-ready version, we will use the allowed extra page to discuss that architecture in that more modular way and we will also include an additional sampling scheme described below as ``on-path-flipping'' (which, like the other sampling scheme, satisfies all our theorems).

    \begin{figure}[H]\centering
        \begin{tikzpicture}[xscale=.85,yscale=.8]\small
          \node[draw,semithick,align=center,text width=2.0cm,minimum height=1cm] (G) at (0,0) {Gradient estimator};
          \node[draw,semithick,align=center,text width=1.0cm,minimum height=1cm] (Q) at (3,0) {$\cR_Q$};
          \node[draw,semithick,align=center,text width=2.0cm,minimum height=1cm] (S) at (8.2,0) {Strategy sampler};
          \node[draw,semithick,align=center,text width=1.6cm,minimum height=.8cm] (E) at (5.5,1.8) {Explorer};
          \node[] (plus) at (5.5,0) {$+$};
          \draw[thick,dashed,blue] (-1.7,-.8) rectangle (4.2,1);
          \node[above left,blue] at (4.2,1) {$\tilde{\cR}$};
          \draw[thick,dashed,purple] (-1.9,-1) rectangle (10,2.7);
          \node[below right,purple] at (10,2.7) {$\cR$};

          \draw[semithick] (5.5, 0) circle (.2);
          \draw[semithick,->] (E.south) --node[right,fill=white]{$\vec{\xi}^t$} (5.5,.2);
          \draw[semithick,->] (-2.5,0) --node[pos=.1,above,fill=white]{$(\vec{\ell}^t)^\top \vec{\pi}^t$} (G.west);
          \draw[semithick,->] (G.east) --node[above,fill=white]{$\tilde{\vec{\ell}}^t$} (Q.west);
          \draw[semithick,->] (Q.east) --node[above,fill=white,inner xsep=0]{$\vec{y}^t\in Q$} (5.3,0);
          \draw[semithick,->] (5.7,0) --node[above,fill=white]{$\vec{w}^t$} (S.west);
          \draw[semithick,->] (S.east) --node[above,fill=white,pos=.55,inner xsep=0]{$\vec{\pi}^t\in\Pi$} (11.8,0);
        \end{tikzpicture}
        \caption{Conceptual construction of our algorithm for the interactive bandit online learning setting.}\label{fig:analysis app}
    \end{figure}
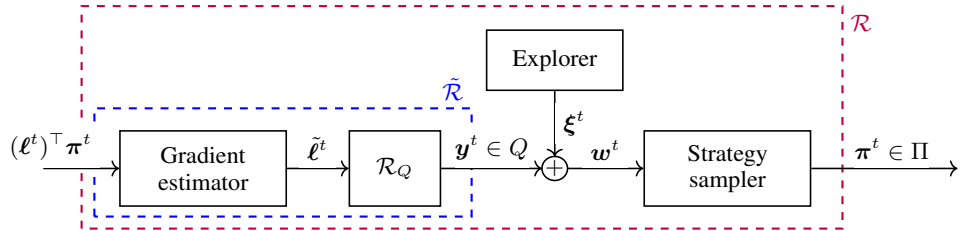

    At a high level, the construction of our interactive bandit regret minimizer works as follows.  At each iteration $t$, we use a full-feedback regret minimizer $\cR_Q$ to output a recommendation for the next sequence-form strategy $\vec{y}^t \in Q$ to be played. Then, we introduce a bias on $\vec{y}^t$, which can be thought of as an exploration term. The resulting biased strategy is $\vec{w}^t \in Q$. Finally, we sample a deterministic policy $\vec{\pi}^t \in \Pi$ starting from $\vec{w}^t$. After the playthrough is over and a terminal node $z^t \in Z$ has been reached, we use the feedback (that is, the terminal node $z^t$ reached in the playthrough, together with its utility $u^t(z^t)$), to construct an unbiased estimator $\tilde{\vec{\ell}}^t$ of the underlying gradient vector that was chosen by the environment.

The algorithm that we have just described accumulates regret from the following three sources. 
\begin{enumerate}
    \item First, it cumulates regret by the full-feedback regret minimizer $\cR_Q$. This scales with the norm of the gradient estimators $\tilde{\vec{\ell}}^t$ as $O(\sqrt{T}\cdot \max_{t=1}^T \|\tilde{\vec{\ell}}^t\| )$.
        \item Second, there is a degradation term due to the fact that we do not exactly follow the recommendations from $\cR_Q$, and instead bias the recommendations. The biasing is necessary to keep the norm of the gradient estimators under control. 
            \item Third, there is a regret degradation term due to the fact that we sample a pure strategy starting from the biased recommendation $\vec{w}^t$.
\end{enumerate}

    Our analysis is split into two main conceptual steps. First, we will quantify the regret degradation that is incurred in passing from $\tilde{\cR}$ to $\cR$ (\cref{fig:analysis}). Specifically, we will study how the regret cumulated by $\tilde{\cR}$,
    \[
        \tilde{R}^T(\vec{\pi}) = \sum_{t=1}^T (\vec{\ell}^t)^\top(\vec{\pi} - \vec{y}^t)
    \]
    is related to the regret cumulated by our overall algorithm,
    \[
        R^T(\vec{\pi}) = \sum_{t=1}^T (\vec{\ell}^t)^\top (\vec{\pi} - \vec{\pi}^t).
    \]
    In this case, the degradation term comes from the fact that the strategy output by $\cR$, that is, $\vec{\pi}^t \in \Pi$, is not the same as the one, $\vec{y}^t$, that was recommended by $\tilde{\cR}$, because an exploration term was added.

The second step in the analysis will be the quantification of the regret degradation that is incurred in passing from $\cR_Q$ to $\tilde{\cR}$ (\cref{fig:analysis}). Specifically, we will study how the regret cumulated by $\tilde{\cR}$ is related to the regret
    \[
        \cR_Q^T(\vec{\pi}) = \sum_{t=1}^T (\tilde{\vec{\ell}}^t)^\top(\vec{\pi} - \vec{y}^t),
    \]
    which is guaranteed to be $O(T^{1/2})$ with high probability. By summing the degradation bounds, the asymptotic regret guarantee will change from $T^{1/2}$ to $T^{3/4}$, while still retaining a high-probability bound.

    Below, we give details about each of the components of our algorithm.

    \paragraph{Gradient Estimator}
    The gradient estimator that we use in our algorithm is given in \cref{lem:gradient estimator}. It is a form of importance-sampling estimator that can be constructed starting from the feedback received by the regret minimizer (that is, the terminal leaf $z^t$ and its corresponding payoff $u^t(z^t)$). It is a particular instantiation of the \emph{outcome sampling} estimator that appeared in the works by~\citet{Lanctot09:Monte}. We follow the formalization of~\citet{Farina20:Stochastic}.

    \begin{lemma}\label{lem:gradient estimator}
        Let $z^t$ be the terminal state in which the process ends at iteration $t$, that is the terminal state reached when the decision maker plays according to the pure strategy $\vec{\pi}^t$ and the environment plays according to the (partially hidden) pure strategy $\vec{y}^t$. Also, let $u^t(z^t)$ be the reward of the decision maker upon termination and let $\sigma(z^t)$ denote the last sequence (decision node-action pair) on the path from the root of the decision process to $z$. Then,
        \[
            \tilde{\vec{\ell}}^t \defeq \frac{u^t(z^t)}{w^t[\sigma(z^t)]}\vec{e}_{\sigma(z^t)}
        \]
        is an unbiased estimator of $\vec{\ell}^t$, that is, $\bbE_t[\tilde{\vec{\ell}}^t] = \vec{\ell}^t$.
    \end{lemma}
    \begin{proof}
      We will prove the lemma by showing that $\bbE_t[(\tilde{\vec{\ell}}^t)^{\!\top} \vec{x}] = (\vec{\ell}^t)^{\!\top} \vec{x}$. 
      %First, as already noticed, $\vec{w}^t$ is known given all the choices made by our algorithm up to time $t-1$ (more formally, $\vec{w}^t$ is $\cF$-measurable). Second, b
      By definition, $z^t$ is the (unique) state such that
        \[
            \pi^t[\sigma(z)] = \prod_{(k,s) \leadsto z} \bbone[s^t_k = s] =1.
        \]
        For all other terminal states $z$, $\pi^t[\sigma(z)]\cdot \prod_{(k,s) \leadsto z} \bbone[s^t_k = s] = 0$. Hence,
        \begin{align*}
            \bbE_t\mleft[(\tilde{\vec{\ell}}^t)^{\!\top} \vec{x}\mright] &= \bbE_t\mleft[\frac{u^t(z^t)}{w^t[\sigma(z^t)]}(\vec{e}_{z}^{\top}\vec{x})\mright]\\
            &= \bbE_t\mleft[\sum_{z \in Z} \mleft(\pi^t[\sigma(z)]\cdot \prod_{(k,s) \leadsto z} \bbone[s^t_k = s]\mright) \frac{u^t(z)}{w^t[\sigma(z)]}(\vec{e}_{z}^{\top}\vec{x})\mright]\\
            &= \sum_{z \in Z}\frac{\bbE_t\big[\pi^t[\sigma(z)]\big]}{w^t[\sigma(z)]} u^t(z)\, \mleft(\prod_{(k,s) \leadsto z} \bbone[s^t_k = s]\mright) \, x[\sigma(z)],
        \end{align*}
        Using the hypothesis that $\vec{\pi}^t$ is a (conditionally) unbiased estimator of $\vec{w}^t$ we further obtain
        \[
            \bbE_t\big[\pi^t[\sigma(z)]\big] = \bbE_t[\vec{\pi}^t][\sigma(z)] = w^t[\sigma(z)]
        \]
        and therefore
        \[
            \bbE_t\mleft[(\tilde{\vec{\ell}}^t)^{\!\top} \vec{x}\mright] = \sum_{z \in Z} u^t(z)\mleft(\prod_{(k,s) \leadsto z} \bbone[s^t_k = s]\mright) \, x[\sigma(z)] = (\vec{\ell}^t)^{\!\top} \vec{x},
        \]
        where we used the definition of $\vec{\ell{}}^t$ from \eqref{eq:ib loss} in the last equality.
    \end{proof}

    \paragraph{The Full-Feedback Regret Minimizer}
    Our full-feedback regret minimizer $\cR_Q$ is the counterfactual regret minimization (CFR) regret minimizer~\cite{Zinkevich07:Regret}. CFR
decomposes the task of computing a strategy $\vec{y}^t \in Q$ for the whole decision process into smaller subproblems at each decision point. The local strategies are then computed via $|\cJ|$ independent full-feedback regret minimizers $\cR_j$ ($j\in \cJ$), one per decision node, that at each $t$ observe a specific feedback---the \emph{counterfactual gradient}---that guarantees low global regret across the decision process. CFR guarantees $O(T^{1/2})$ regret in the worst case against any strategy $\hat{\vec{y}}^t \in Q$.

We now review how the counterfactual gradients that are given as feedback at all local regret minimizers $\cR_j$ are constructed. Let $\vec{x}_j^t$, $j\in \cJ$, be the local strategies output by the $\cR_j$'s, and let $\tilde{\vec{\ell}}^t$ be the gradient vector observed at time $t$ by $\cR_Q$. We start by constructing the counterfactual values $V_j$, indexed over $j\in\cJ$, which are recursively defined as
    \[
        V^t_j = \sum_{a \in A_j} x^t_j[a] \mleft(\tilde{{\ell}}^t[j,a] + \sum_{j'\in \cJ: p_{j'} = (j,a)} V^t_{j'}\mright).
    \]
    Then, for all $j\in\cJ$ we construct the counterfactual gradient $\ell^t_j$ by setting $\ell_j^t[a] = \tilde{{\ell}}^t[j,a] + \sum_{j' \in\cJ : p_{j'} = (j,a)} V^t_{j'}$ for all $a \in A_j$.

    Let $R_j^T(\hat{\vec{\pi}}_j)$ denote the regret cumulated up to time $T$ by the local full-feedback regret minimizer $\cR_j$ compared to strategy $\hat{\vec{\pi}}_j^t \in \Delta^{|A_j|}$. Then, the regret cumulated by the CFR algorithm is known to satisfy the regret bound in \cref{lem:RQ}.
    \begin{lemma}[\cite{Farina19:Online}]\label{lem:RQ}
        At all $T$, for all $\hat{\vec{\pi}} \in Q$,
        \[
            R_Q(\hat{\vec{\pi}}) \le \max_{\vec{q}\in Q} \mleft\{\sum_{j\in \cJ} q[p_j]\cdot \max_{\hat{\vec{\pi}}_j \in \Delta^{|A_j|}} R_j^T(\hat{\vec{\pi}}_j)\mright\}.
        \]
    \end{lemma}

    \paragraph{Exploration Term} The role of the exploration term is to reduce the norm of the gradient estimators. We bias the strategies $\vec{y}^t$ output by $\cR_Q$ by taking a convex combination with the \emph{exploration strategy} $\vec{\xi}^t \in Q$ defined as the sequence-form strategy that picks action $a \in A_j$ at decision node $j\in\cJ$ with probability proportional to
\[
    \xi^t[j,a] \propto \frac{h^t(j,a)}{\sum_{a\in A_j} h^t(j,a)},
\]
where $h^t$ is a generic exploration function $\Sigma \to \bbR_{> 0}$.

Let $\sigma(z)$ denote the last sequence (decision node-action pair) encountered on the path from the root of the decision process to terminal node $z \in Z$. A key quantity about $\vec{\xi}^t$, that will be important to quantify the effect of the biasing on the norm of the gradient estimator, is the following:
\[
    \rho^t \defeq \max_{z \in Z} \frac{1}{\xi^t[\sigma(z)]} = \max_{z \in Z}\prod_{(j,a)\leadsto z} \mleft(\frac{\sum_{a' \in A_j} h^t(j,a')}{h^t(j,a)}\mright).
\]

At all times $t$, we construct the biased strategy $\vec{w}^t$ starting from the unbiased strategy $\vec{y}^t$ and the exploration strategy $\vec{\xi}^t$ as
\begin{equation}\label{def:biased estimate}
    \vec{w}^t \defeq (1-\beta^t)\cdot \vec{y}^t + \beta^t \cdot \vec{\xi}^t.
\end{equation}

    We call the exploration function $h^t$ that assigns the constant value $1$ to all sequences the \emph{uniform exploration strategy}. We also study an exploration strategy for which we can prove better worst-case convergence speed. Specifically, we call the exploration function $h^t$ that assigns, to each sequence $(j,a)\in\Sigma$, the number of terminal nodes under $(j,a)$ the \emph{balanced exploration strategy}. For this strategy, the minimum reach probability across all terminal nodes, which is important for our worst-case convergence guarantee, scales linearly in the game size.
    \begin{lemma}[\cite{Farina20:Stochastic}]\label{lem:rho for balanced}
        When $h^t$ is a perfect measure of the number of terminal nodes under each sequence $(j,a) \in \Sigma$, then $\rho^t \le |\Sigma| - 1$.
    \end{lemma}

\paragraph{Strategy Sampler} The role of the strategy sampler is to provide an \emph{unbiased} estimator $\vec{\pi}^t \in \Pi$ of $\vec{w}^t$. In fact, any unbiased estimator can be used in our framework. We start by describing a sampling scheme for sequence-form strategies that is folklore in the literature. In the natural sampling scheme, given a sequence-form strategy $\vec{q} \in Q$, at all decision nodes $j \in \cJ$ an action is sampled according to the distribution $q[j,a]/q[p_j] : a \in A_j$.

\begin{itemize}
    \item \emph{Upfront-flipping strategy sampling}
Since by construction $\vec{w}^t = (1-\beta^t)\cdot \vec{y}^t + \beta^t \cdot \vec{\xi}^t$, sampling from $\vec{w}^t$ can be done as follows. First, we flip a biased coin, where the probability of heads is $1-\beta^t$ and the probability of tails is $\beta^t$. If heads comes up, we sample actions according to the natural sampling scheme applied to the sequence-form strategy $\vec{y}^t$ output by $\cR_Q$. Otherwise, we sample actions according to the natural sampling scheme applied to the exploration strategy $\vec{\xi}^t$. We call the strategy sampler just described the \emph{upfront-flipping} sampling scheme.

    \item \emph{On-path-flipping strategy sampling} We now describe a different sampling scheme, which we coin the \emph{on-path-flipping} sampling scheme. First, we compute the strategy $\vec{w}^t$ by explicitly taking the convex combination between $\vec{y}^t$ and the exploration term $\vec{\xi}^t$. Then we apply the natural sampling scheme for sequence-form strategies to $\vec{w}^t$. In the on-path-flipping sampling scheme, the exploration and the exploitation are interleaved at each decision node in a counterintuitive way.
\end{itemize}
Because both sampling schemes are unbiased, our theory, without changes, applies to both of these sampling schemes. Later in this appendix, we report experiments on both of these sampling schemes. 

    \subsection{Relationship Between Exploration and Norm of Counterfactual Gradients}\label{sec:sparse update}

    The gradient $\tilde{\vec{\ell}}^t$ that is given to $\cR_Q$ as feedback is given in \cref{lem:gradient estimator}. It has zero entries everywhere, except for the sequence $\sigma(z^t)$. Consequently, the counterfactual values $V_j^t$ are $0$ everywhere, except for the decision nodes that were traversed on the path from the root to the terminal node $z^t \in Z$. In turn, this means that all counterfactual gradients are $0$, except for the sequences that are traversed. for them, the construction of the counterfactual gradients reduces to the \emph{regret update phase} described in \cref{sec:algorithm}.

    The norms of the counterfactual gradients are maximum at the leaves. In particular, the following is known.
    \begin{lemma}[\cite{Lanctot09:Monte,Farina20:Stochastic}]
        Let $\Delta$ be the maximum range of payoffs that can be selected by the environment at all times. All counterfactual gradients have norm upper bounded as
        \[
            \|\vec{\ell}_j^t\|_2 \le \frac{\Delta\rho^t}{\beta^t} \quad\forall j\in\cJ.
        \]
        Furthermore, the loss estimate $\tilde{\vec{\ell}}^t$ satisfies
        \[
            (\tilde{\vec{\ell}}^t)^\top (\vec{x} - \vec{x}') \le \frac{\Delta \rho^t}{\beta^t} \quad \forall \vec{x},\vec{x}' \in Q.
        \]
    \end{lemma}

    \subsection{Relationship Between $\cR$ and $\tilde{\cR}$}

    As a first step in our analysis, we establish an important relationship between the regret cumulated by $\cR$ and $\tilde{\cR}$. Fundamentally, it quantifies the degradation in the regret incurred from playing $\vec{\pi}^t$ instead of $\vec{y}^t$ as recommended. The degradation is kept under control by the fact that the expectation of the output $\vec{\pi}^t$ is close to $\vec{y}^t$ when $\beta$ is small.
    Before we state the central result of this subsection (\cref{prop:relationship R Rtilde}), we need the following simple bound.

    \begin{lemma}\label{lem:gap sum}
        Let $\Delta$ be the payoff range of the interaction, that is the maximum absolute value of any payoff that can be selected by the environment at any time. At all $T$, it holds that
        \[
            \sum_{t=1}^T (\vec{\ell}^t)^\top (\vec{y}^t - \vec{w}^t) \le \Delta \sum_{t=1}^T \beta^t.
        \]
    \end{lemma}
    \begin{proof}Using simple algebraic manipulations,
        \begin{align*}
            \sum_{t=1}^T (\vec{\ell}^t)^\top (\vec{y}^t - \vec{w}^t)
              &= \sum_{t=1}^T (\vec{\ell}^t)^\top \mleft(\vec{y}^t - \big((1-\beta^t)\vec{y}^t + \beta^t \vec{b}^t \big)\mright) \\
              &= \sum_{t=1}^T \beta^t(\vec{\ell}^t)^\top (\vec{y}^t - \vec{b}^t) \le \Delta \sum_{t=1}^T \beta^t,
        \end{align*}
        where the last inequality follows from the definition of $\Delta$.
    \end{proof}

    With \cref{lem:gap sum} we are ready to analyze the relationship between $\cR$ and $\tilde{\cR}$.

    \begin{proposition}\label{prop:relationship R Rtilde app}
        At all $T$, for all $p\in(0,1)$ and $\vec{\pi} \in \Pi$, with probability at least $1-p$,
        \[
            R^T(\vec{\pi}) \le \tilde{R}^T(\vec{\pi}) +  \Delta\mleft( \sqrt{2T\log\frac{1}{p}}+ \sum_{t=1}^T \beta^t\mright).
        \]
    \end{proposition}
    \begin{proof}
        Introduce the discrete-time stochastic process
        \[
            d^t \defeq ({\vec{\ell}}^t)^{\!\top} (\vec{\pi}^t - \vec{w}^t), \quad t \in \{1, 2,\dots\}.
        \]
        From the online learning model hypotheses, we have that ${\vec{\ell}}^t$ is conditionally independent from $\vec{\pi}^t$ and $\vec{w}^t$ given all past choices of the algorithm up to time $t-1$, and since $\bbE_t[\vec{\pi}^t] = \vec{w}^t$ by construction, $d^t$ is a martingale difference sequence. Furthermore, some elementary algebra reveals that
        \begin{align*}
            \sum_{t=1}^T d^t &= \sum_{t=1}^T ({\vec{\ell}}^t)^\top (\vec{\pi}^t - \vec{w}^t)\\
                &= \sum_{t=1}^T ({\vec{\ell}}^t)^\top (\vec{\pi} - \vec{y}^t) - \sum_{t=1}^T ({\vec{\ell}}^t)^\top (\vec{\pi} - \vec{\pi}^t) + \sum_{t=1}^T ({\vec{\ell}}^t)^\top (\vec{y}^t - \vec{w}^t) \\
                &= \tilde{R}^T(\vec{\pi}) - R^T(\vec{\pi}) + \sum_{t=1}^T ({\vec{\ell}}^t)^\top (\vec{y}^t - \vec{w}^t) \\
                &\le \tilde{R}^T(\vec{\pi}) - R^T(\vec{\pi}) + \Delta \sum_{t=1}^T \beta^t, \numberthis{eq:sum dt 1}
        \end{align*}
        where the last inequality follows from \cref{lem:gap sum}.
        Since $|\delta^t| \le \Delta$ for all $t$, using the Azuma-Hoeffding inequality~\cite{Hoeffding63:Probability,Azuma67:Weighted} we have that
        \begin{align*}
            1-p &\le \bbP\mleft[\sum_{t=1}^T d^t \ge -\Delta\sqrt{2T\log\frac{1}{p}}\mright]\\
                &\le \bbP\mleft[\tilde{R}^T(\vec{\pi}) - R^T(\vec{\pi}) + \Delta \sum_{t=1}^T \beta^t \ge -\Delta\sqrt{2T\log\frac{1}{p}}\mright]\\
                &= \bbP\mleft[R^T(\vec{\pi}) \le \tilde{R}^T(\vec{\pi}) + \Delta \mleft(\sqrt{2T\log\frac{1}{p}} + \sum_{t=1}^T \beta^t\mright)\mright],
        \end{align*}
        where the second inequality used \eqref{eq:sum dt 1}.
    \end{proof}

    \subsection{Relationship Between $\tilde{\cR}$ and $\cR_Q$}

    The next proposition establishes the important relationship between the regret cumulated by $\tilde{\cR}$ and $\cR_Q$. Unlike \cref{prop:relationship R Rtilde}, where a regret degradation was suffered for playing a strategy different that the recommended one, in this case the regret degradation comes from the fact that the gradient that is observed by $\cR_Q$ is different from that observed by $\tilde{\cR}$. However, as we will show the degradation is kept under control by the fact that $\tilde{\vec{\ell}}^t$ is an unbiased estimator of $\vec{\ell}^t$ by hypothesis.
    \begin{proposition}\label{prop:relationship Rtilde RQ app}
        At all $T$, for all $p\in(0,1)$ and $\vec{\pi} \in \Pi$, with probability at least $1-p$,
        \[
            \tilde{R}^T(\vec{\pi}) \le R_Q^T(\vec{\pi}) + \frac{\Delta}{\beta^T}\nu \,\sqrt{2T\log\frac{1}{p}},
        \]
        where
        \[
            \nu \defeq \sqrt{\frac{1}{T}\sum_{t=1}^T (\rho^t)^2} = \sqrt{\frac{1}{T}\sum_{t=1}^T \max_{z\in Z} \mleft\{ \prod_{(j,a) \leadsto z} \mleft(\frac{\sum_{a'\in A_{j}} h^t(j, a')}{h^t(j,a)}\mright)^{\!\!2}\mright\}}.
        \]
    \end{proposition}
    \begin{proof}
        Introduce the discrete-time stochastic process
        \[
            d^t \defeq (\vec{\ell}^t - \tilde{\vec{\ell}}^t)^{\!\top} (\vec{\pi} - \vec{y}^t), \quad t \in \{1,2,\dots\}.
        \]
        From the online learning model hypotheses, we have that $\vec{\ell}^t$ and $\tilde{\vec{\ell}}^t$ are conditionally independent from $\vec{y}^t$, and since $\bbE_t[\tilde{\vec{\ell}}^t]=\vec{\ell}^t$ by construction, $d^t$ is a martingale difference sequence. At each $t$, the conditional range of $d^t$ is upper bounded by
        \[
            |d^t| \le |(\vec{\ell}^t)^\top (\vec{y}^t - \vec{\pi})| + |(\tilde{\vec{\ell}}^t)^\top (\vec{y}^t - \vec{\pi})| \le \Delta + \frac{\Delta}{\beta^t}\rho^t \le \frac{2\Delta}{\beta^t}\rho^t,
        \]
        where the last inequality follows since $\beta^t \le 1$ and $\rho^t \ge 1$.
        Hence, the sum of quadratic ranges for the martingale difference sequence is
        \begin{align*}
            r &\defeq \sqrt{\frac{1}{T}\sum_{t=1}^T \mleft(\frac{2\Delta}{\beta^t}\rho^t\mright)^{\!2}}
              \le \sqrt{\frac{1}{T}\sum_{t=1}^T \mleft(\frac{2\Delta}{\beta^T}\rho^t\mright)^{\!2}} = \frac{2\Delta}{\beta^T}\nu,
        \end{align*}
        where we used the fact that the $\beta^t$ are (weakly) decreasing.
        Using the Azuma-Hoeffding inequality, we obtain
        \begin{align*}
            1-p &\le \bbP\mleft[\sum_{t=1}^T d^t \le r\,\sqrt{\frac{T}{2}\log\frac{1}{p}}\mright]\\
                &= \bbP\mleft[\sum_{t=1}^T d^t \le \frac{\Delta}{\beta^T} \nu \,\sqrt{2T\log\frac{1}{p}}\mright].\numberthis{eq:azuma 2}
        \end{align*}
        Finally,
        \begin{align*}
            \sum_{t=1}^T d^t &= \sum_{t=1}^T (\vec{\ell}^t - \tilde{\vec{\ell}}^t)^{\!\top} (\vec{y}^t - \vec{\pi})\\
                &= \sum_{t=1}^T (\vec{\ell}^t)^{\!\top} (\vec{y}^t - \vec{\pi}) - \sum_{t=1}^T (\tilde{\vec{\ell}}^t)^{\!\top} (\vec{y}^t - \vec{\pi})\\
                &= \tilde{R}^T(\vec{\pi}) - R_Q^T(\vec{\pi}),\numberthis{eq:sum dt 2}
        \end{align*}
        and substituting \eqref{eq:sum dt 2} into \eqref{eq:azuma 2} we have
        \[
            \bbP\mleft[\tilde{R}^T(\vec{\pi}) - R_Q^T(\vec{\pi}) \le \frac{\Delta}{\beta^T}\nu \,\sqrt{2T\log\frac{1}{p}}\mright] \ge 1-p.
        \]
        Rearranging the terms inside of the square brackets yields the statement.
    \end{proof}

    \subsection{Regret Analysis for the Overall Algorithm}

    Combining \cref{prop:relationship Rtilde RQ} and \cref{prop:relationship R Rtilde} using the union bound lemma, we obtain the following (note that the fractions $1/p$ inside of the logarithms have become $2/p$ as a consequence of the union bound lemma).
    \begin{corollary}
        At all $T$, for all $p\in(0,1)$ and $\vec{\pi} \in \Pi$, with probability at least $1-p$,
        \[
            R^T\!(\vec{\pi}) \le R_Q^T(\vec{\pi}) + \frac{\Delta}{\beta^T} (1 + \nu) \sqrt{2T \log\frac{2}{p}} + \Delta \sum_{t=1}^T \beta^t.
        \]
    \end{corollary}
    where $\nu$ is as in \cref{prop:relationship Rtilde RQ app}.
    \begin{proof}
        From \cref{prop:relationship Rtilde RQ} and \cref{prop:relationship R Rtilde}, respectively, we have that
        \begin{align*}
            \frac{p}{2} &\ge \bbP\mleft[R^T(\vec{\pi}) \ge \tilde{R}^T(\vec{\pi}) + \Delta\mleft(\sqrt{2T\log\frac{2}{p}} + \sum_{t=1}^T \beta^t\mright)\mright]\\
            \frac{p}{2} &\ge \bbP\mleft[\tilde{R}^T(\vec{\pi}) \ge R_Q^T(\vec{\pi}) + \frac{\Delta}{\beta^T}\nu\,\sqrt{2T\log\frac{2}{p}}\mright].
        \end{align*}
        Summing the two inequalities and using the union bound, we obtain
        \begin{align*}
            p &\ge \bbP\mleft[R^T(\vec{\pi}) \ge \tilde{R}^T(\vec{\pi}) + \Delta \mleft(\sqrt{2T\log\frac{2}{p}} + \sum_{t=1}^T \beta^t\mright)\mright] + \bbP\mleft[\tilde{R}^T(\vec{\pi}) \le R_Q^T(\vec{\pi}) + \frac{\Delta}{\beta^T}\nu\sqrt{2T\log\frac{2}{p}}\mright]\\
             &\ge \bbP\mleft[\mleft(R^T(\vec{\pi}) \ge \tilde{R}^T(\vec{\pi}) + \Delta \mleft(\sqrt{2T\log\frac{2}{p}} + \sum_{t=1}^T \beta^t\mright)\mright) \lor \mleft(\tilde{R}^T(\vec{\pi}) \le R_Q^T(\vec{\pi}) + \frac{\Delta}{\beta^T}\nu\sqrt{2T\log\frac{2}{p}}\mright)\mright]\\
             &\ge \bbP\mleft[
                R^T(\vec{\pi})
                +
                \tilde{R}^T(\vec{\pi}) \ge \mleft(\tilde{R}^T(\vec{\pi}) + \Delta \mleft(\sqrt{2T\log\frac{2}{p}} + \sum_{t=1}^T \beta^t\mright)\mright) + \mleft(R_Q^T(\vec{\pi}) + \frac{\Delta}{\beta^T}\nu\sqrt{2T\log\frac{2}{p}}\mright)\mright]\\
            &= \bbP\mleft[\R^T\!(\vec{\pi}) \ge R_Q^T(\vec{\pi}) + \frac{\Delta}{\beta^T}(\beta^T + \nu) \sqrt{2T \log\frac{2}{p}} + \Delta \sum_{t=1}^T \beta^t\mright]\\
            &\ge \bbP\mleft[\R^T\!(\vec{\pi}) \ge R_Q^T(\vec{\pi}) + \frac{\Delta}{\beta^T}(1 + \nu) \sqrt{2T \log\frac{2}{p}} + \Delta \sum_{t=1}^T \beta^t\mright].
        \end{align*}
        Taking complements yields the statement.
    \end{proof}

    \thmregretspecific*
        \begin{proof}Fix $T \in \{1,2,\dots\}$. At all times $t \in \{1,\dots,T\}$, the norm that enter regret minimizer $\cR_j$ has norm upper bounded by

            \[
                \|\vec{\ell}^t_j\| \le \frac{\Delta}{\beta^t}\rho^t \le
                  \Delta \max_{t=1}^T\{\rho^t \} \max\mleft\{\frac{1}{k},1\mright\} T^{1/4},
            \]
            where we used the fact that $\beta^T \ge \min\{k,1\}T^{-1/4}$ for all $T$. Letting $\tilde{M} \defeq \max_{t=1}^T \rho^t$ and $\tilde{k} \defeq \max\mleft\{\frac{1}{k},1\mright\}$, we obtain $\|\vec{\ell}_j^t\| \le \Delta \tilde{M}\tilde{k} T^{1/4}$.

            Let $c_j$ be the constant in the regret guarantee for the $O(T^{1/2})$ regret of $\cR_j$. Since regret guarantees are always with respect to some ball of gradients of bounded norm (say, norm $1$) the regret cumulated by $\cR_j$ is bounded as $\max_{\hat{\vec{\pi}}_j} R_j^T(\hat{\vec{\pi}}_j) \le c_j \|\vec{\ell}^t_j\| T^{1/2}$. In other words, we need to keep into account the degradation factor due to the fact that the norms of the gradients constructed through \cref{lem:gradient estimator} might exceed the bound for which the regret guarantee for $\cR_j$ was given. So, in particular,
            \[
                \max_{\hat{\vec{\pi}}_j} R_j^T(\hat{\vec{\pi}}_j) \le c_j \mleft(\Delta\tilde{M}\tilde{k} T^{1/4}\mright) T^{1/2} = \Delta \tilde{M} \tilde{k} c_j T^{3/4},
            \]
            and so
            \begin{equation}\label{eq:bound1}
                \max_{\vec{q} \in Q}\mleft\{\sum_{j \in \cJ} q[p_j]\cdot \max_{\hat{\vec{\pi}}_j} R_j^T(\hat{\vec{\pi}}_j)\mright\} \le \sum_{j\in\cJ} \max\left\{0, \max_{\hat{\vec{\pi}}_j} R_j^T(\hat{\vec{\pi}}_j)\right\} \le {\Delta}\tilde{M}\tilde{k}\mleft(\sum_{j\in\cJ}c_j\mright)\, T^{3/4}.
            \end{equation}
            Furthermore,
            \begin{equation}\label{eq:bound2}
                \frac{\Delta}{\beta^T}(1+\nu)\sqrt{2T\log\frac{2}{p}} \le \Delta\tilde{k}T^{1/4}(1+\nu)\sqrt{2T\log\frac{2}{p}} = \Delta (1+\nu)\tilde{k} T^{3/4} \sqrt{\log\frac{2}{p}}.
            \end{equation}
            Finally,
            \begin{equation}\label{eq:bound3}
                \Delta\sum_{t=1}^T \beta^t = \Delta \sum_{t=1}^T\min\{1, k \cdot t^{-1/4}\} \le \Delta k \sum_{t=1}^T t^{-1/4} \le \Delta k \int_{0}^{T} t^{-1/4} dt = \frac{4\Delta k}{3} T^{3/4}.
            \end{equation}
            Substituting the bounds~\eqref{eq:bound1},~\eqref{eq:bound2}, and~\eqref{eq:bound3} into the general result of \cref{thm:regret bound general}, we obtain
           \begin{align*}
             R^T(\vec{\pi}) &\le \Delta\tilde{M}\tilde{k}\mleft(\sum_{j\in\cJ}c_j\mright)\, T^{3/4} + \Delta (1+\nu)\tilde{k} T^{3/4} \sqrt{\log\frac{2}{p}} + \frac{3\Delta k}{4} T^{3/4}\\
                &= {\Delta}\mleft(\frac{4k}{3} + (1+\nu)\tilde{k}\sqrt{\log\frac{2}{p}} + \tilde{k}\tilde{M}\sum_{j\in\cJ}c_j \mright) T^{3/4}\\
                &\le {\Delta}\mleft(\frac{4k}{3\log 2} + (1+\nu)\tilde{k} + \frac{\tilde{k}}{\log 2}\tilde{M}\sum_{j\in\cJ}c_j \mright) T^{3/4}\sqrt{\log\frac{2}{p}},
           \end{align*}
           where we used the fact that $\log(2/p) \ge \log 2$ for all $p \in (0,1)$. Setting $c$ to be the quantity in parentheses and noting that when $h^t$ measures the number of terminal states in each subtree $\nu \le |\Sigma| - 1$ (\cref{lem:rho for balanced}), we obtain the statement.
        \end{proof}

        \begin{remark}
            \cref{thm:regret specific} shows that our algorithm achieves $O(T^{3/4} \sqrt{\log(1/p)}$ regret with high probability (specifically, probability at least $1-p$). Indeed, using the properties of logarithms, $\sqrt{\log(2/p)} = \sqrt{\log(1/p) + \log 2} \le 2\sqrt{\log(1/p)} = O(\sqrt{\log(1/p)})$ for all $p \le 1/2$.
        \end{remark}

    \subsection{Handling the Unknown Structure}

    As already noted in Section~\ref{sec:sparse update}, only the local full-feedback regret minimizers $\cR_j$ on the path from the root to $z^t$ observe a nonzero counterfactual gradient, which is computed as in the \emph{regret update phase} described in \cref{sec:algorithm}. All other local regret minimizers see a zero gradient, and therefore their regret does not increase. So, we can safely avoid updating those regret minimizers that are not on the path from the root to the most recent terminal leaf.

    The upfront-flipping sampling scheme and the on-path-flipping sampling scheme can both be implemented so that the actions are sampled incrementally as the interaction with the environment progresses. In other words, there is no need to sample actions at all decision nodes upfront in order to interact with the environment. In the body of this paper, we only described how to do so for the upfront-flipping sampling scheme, which is conceptually the most important, as it showcases well the difference with, for example, epsilon-greedy exploration used by MCCFR. 

    \section{Additional Details about Experiments}\label{app:games}

\subsection{Description of Games}

Here, we review the two standard benchmark games that we use in the experiments.

In Kuhn poker~\citep{Kuhn50:Simplified}, two players put an ante worth $1$ into the pot at the beginning of the game. Then, each player is privately dealt one card from a deck that contains only three cards---specifically, jack, queen, and king. Then Player $1$ decides to either check or bet $1$. If Player~1 checks, Player 2 can decide to either check or raise $1$. If Player 2 checks, a showdown occurs. If Player 2 raises, Player 1 can fold---and the game ends with Player 2 taking the pot---or call, at which point a showdown occurs. Otherwise, if the first action of Player~1 is to raise, then Player 2 may fold (the game ends with Player~1 taking the pot) or call, at which point a showdown occurs. In a showdown, the player with the higher card wins the pot and the game ends.

\emph{Leduc poker}~\citep{Southey05:Bayes} is played with a deck of 3 unique
ranks, each appearing twice in the deck. There are two rounds in the game. In the
first round, all players put an ante of $1$ in the pot and are privately dealt a single card. A round of betting then starts. Player 1 acts first, and at most two bets are allowed per player. Then, a card is publicly revealed, and another
round of betting takes place, with the same dynamics described above. After the two betting round, if one of the players has a pair with the public card, that player
wins the pot. Otherwise, the player with the higher card wins the pot. All bets in the first
round are worth $2$, while all bets in the second round are $4$.

\subsection{Hyperparameter Tuning for the Policy Gradient Approach by \citet{Srinivasan18:Actor}}

We acknowledge the help of some of the original authors \citep{Srinivasan18:Actor}
in tuning hyperparameters for their algorithm on the Leduc poker benchmark.

For both the RPG and the QPG formulation, the following combinations of hyperparameters was tested:
\begin{itemize}[leftmargin=8mm,itemsep=2mm,nolistsep]
    \item Critic learning rate: $\{0.001, 0.05, 0.01\}$;
    \item Pi learning rate: $\{0.001, 0.01, 0.05\}$;
    \item Entropy cost (used to multiply the entropy loss): $\{0.01, 0.1\}$;
    \item Batch size (for Q and Pi learning): $\{16,64,128\}$;
    \item Num critics (before each Pi update): $\{16,64,128\}$.
\end{itemize}

In \cref{table:hyper1,table:hyper2} we report the average exploitability (computed across 5 independent runs for each choice of hyperparameters) for the top 5 choices of hyperparameters for the RPG and QPG formulation, respectively. The experimental results in the body show data for the best combination of eitherparameters (first row of each table).

\begin{table}[H]\centering
    \scalebox{.95}{\begin{tabular}{rrrrr|r}
        \toprule
        Critic learning rate & Pi learning rate & Entropy cost & Batch size & Num Critics & Avg. exploitability\\
        \midrule
        0.01  & 0.05 & 0.01 & 128 &  64 & 0.448680\\
        0.01  & 0.05 & 0.01 &  64 &  64 & 0.451409\\
        0.001 & 0.01 & 0.01 &  16 & 128 & 0.472729\\
        0.01  & 0.01 & 0.01 & 128 &  16 & 0.493694\\
        0.05  & 0.01 & 0.01 & 128 &  64 & 0.494840\\
        \bottomrule
\end{tabular}}
\caption{Performance of PG with RPG gradient formulation using the top 5 combinations of hyperparameters.}
\label{table:hyper1}
\end{table}

\begin{table}[H]\centering
    \scalebox{.95}{\begin{tabular}{rrrrr|r}
        \toprule
        Critic learning rate & Pi learning rate & Entropy cost & Batch size & Num Critics & Avg. exploitability\\
        \midrule
        0.01  & 0.05 & 0.01 &  64 & 128 & 0.421745\\
        0.01  & 0.01 & 0.01 &  64 &  16 & 0.451535\\
        0.001 & 0.01 & 0.01 &  16 & 128 & 0.457430\\
        0.01  & 0.01 & 0.01 &  64 & 128 & 0.471100\\
        0.05  & 0.01 & 0.01 & 128 &  64 & 0.524005\\
        \bottomrule
\end{tabular}}
\caption{Performance of PG with QPG gradient formulation using the top 5 combinations of hyperparameters.}
\label{table:hyper2}
\end{table}

\subsection{Additional Experiments}

In \cref{fig:kuhn conv,fig:kuhn coin,fig:leduc conv,fig:leduc coin} we study how the uniform and balanced exploration functions compare in Kuhn and Leduc poker, for both the on-path-flipping and the upfront-flipping sampling scheme. We test different values of the exploration multiplier $k$ (\cref{thm:regret specific}), specifically $k \in \{0.5,1,10,20\}$.

The experiments show that the on-path-flipping sampling schemes leads to significantly less variance than the upfront-flipping sampling scheme. Sometimes the latter leads to better average convergence very early on in the learning, but overall has worse convergence in practice.

The experiments show no meaningful different between the uniform and balanced exploration strategies.

\cref{fig:mccfr} shows how MCCFR performed with different exploration parameters. 

\begin{figure}[H]
  \centering
  \includegraphics[scale=.65]{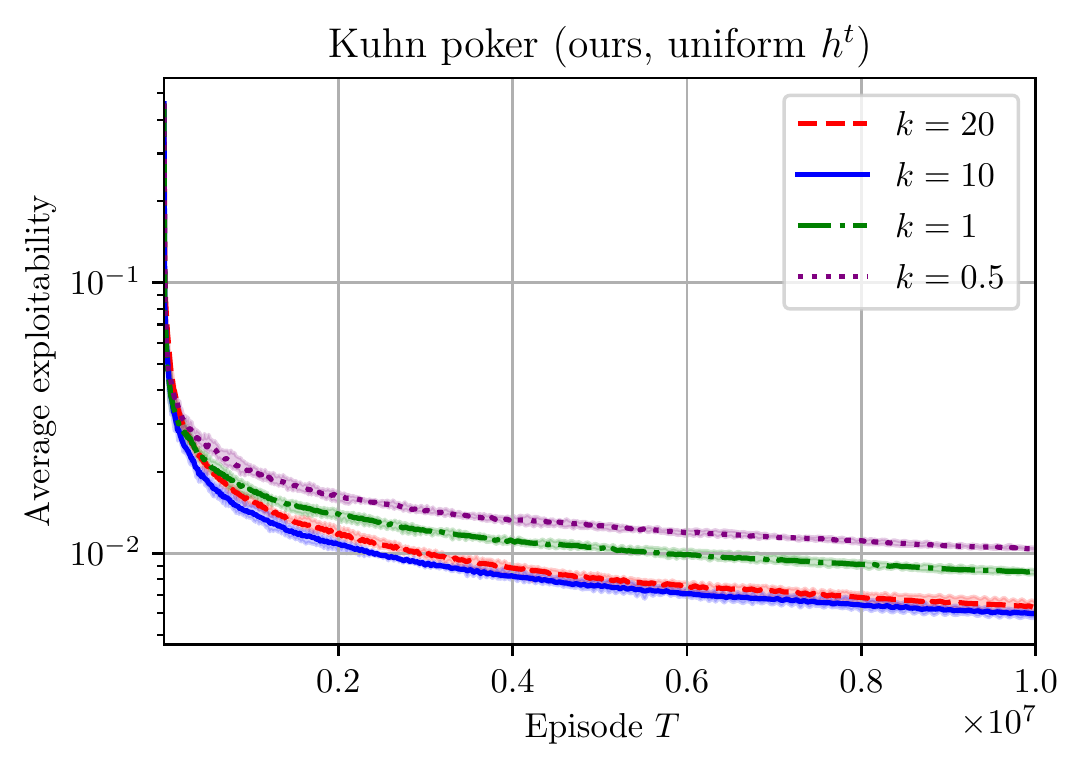}%
  \includegraphics[scale=.65]{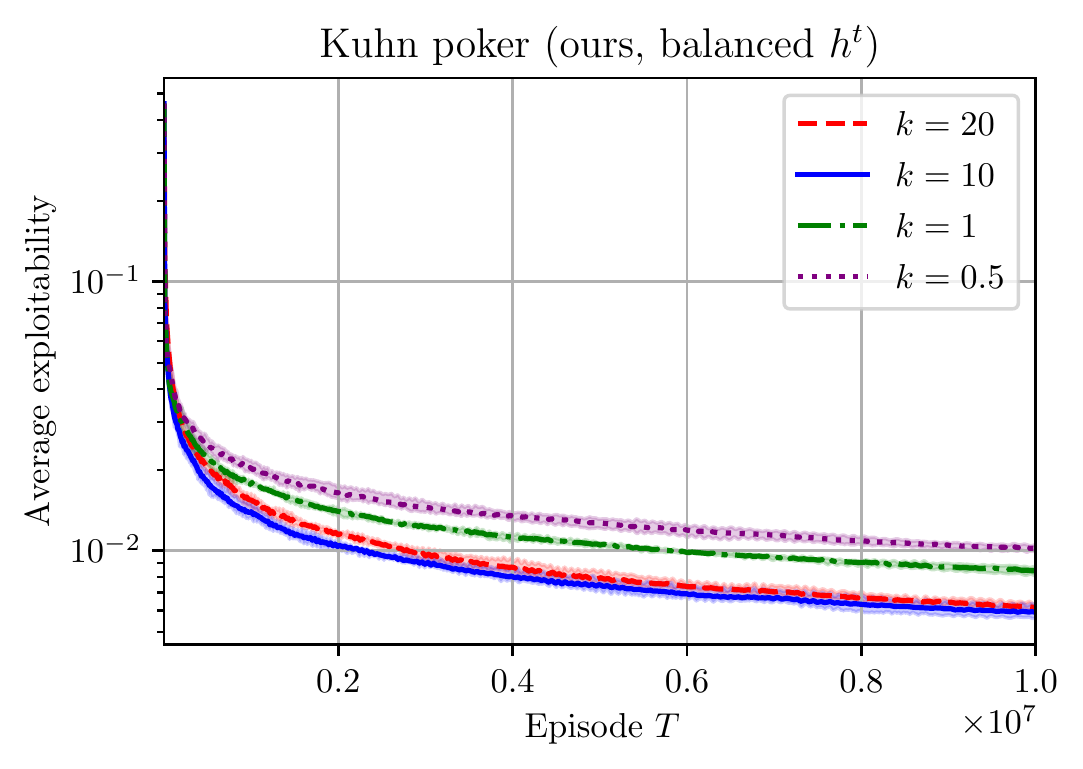}
  \caption{Comparison between uniform and balanced exploration strategies in Kuhn poker, for different values of the exploration multipler $k$, when using the on-path-flipping sampling scheme.}\label{fig:kuhn conv}
\end{figure}

\begin{figure}[H]
  \centering
  \includegraphics[scale=.65]{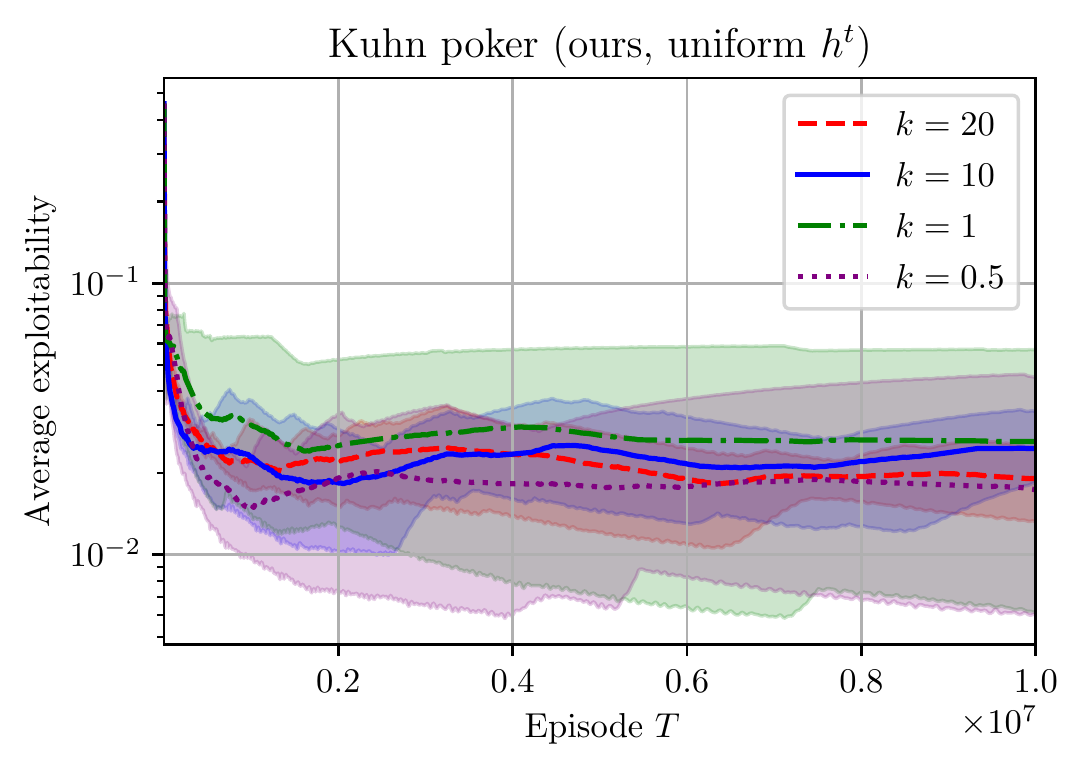}%
  \includegraphics[scale=.65]{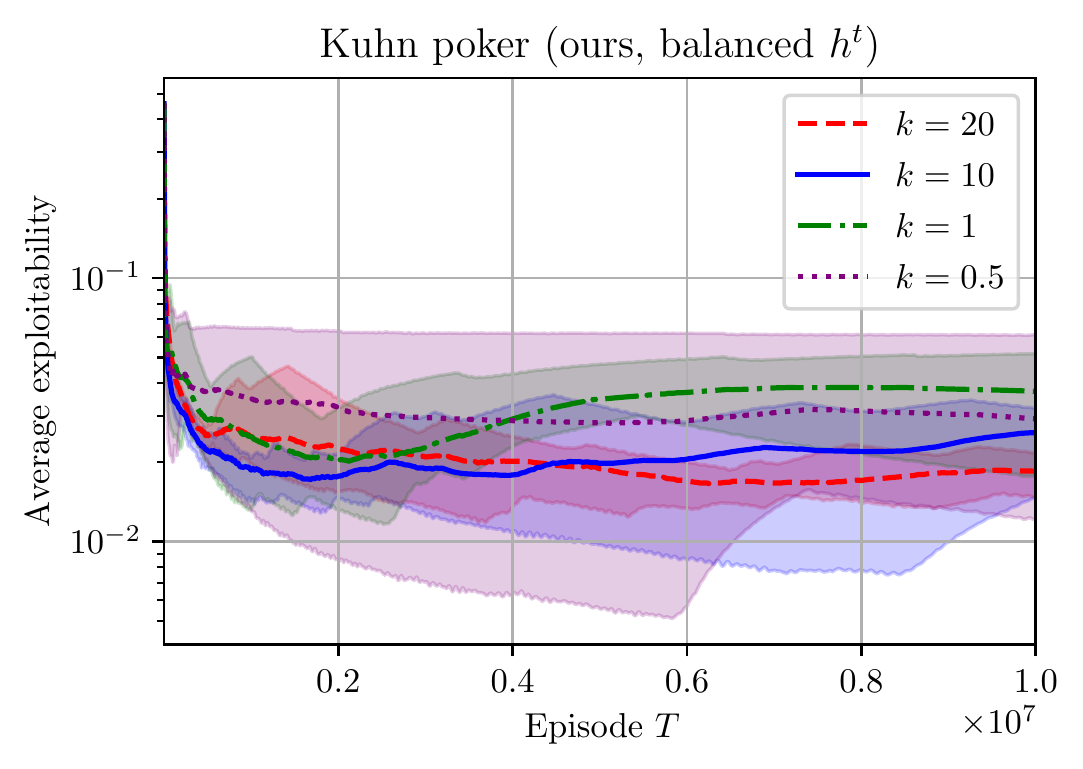}
  \caption{Comparison between uniform and balanced exploration strategies in Kuhn poker, for different values of the exploration multipler $k$, when using the upfront-flipping sampling scheme.}\label{fig:kuhn coin}
\end{figure}
\begin{figure}[H]
  \centering
  \includegraphics[scale=.65]{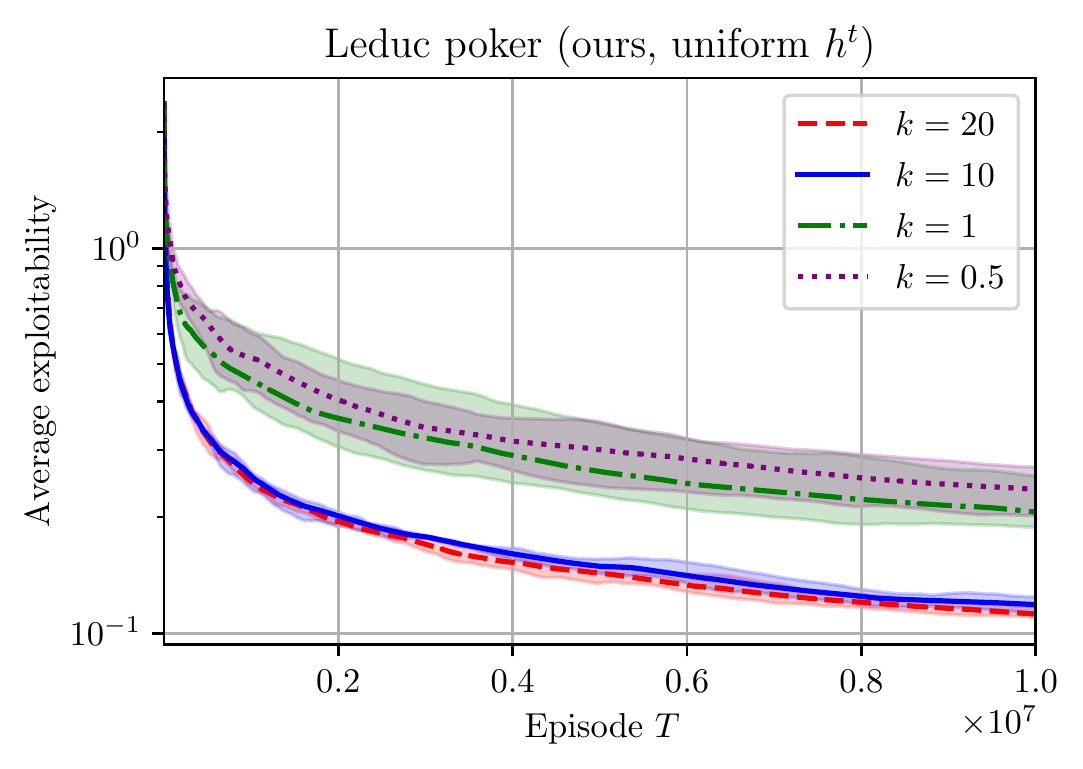}%
  \includegraphics[scale=.65]{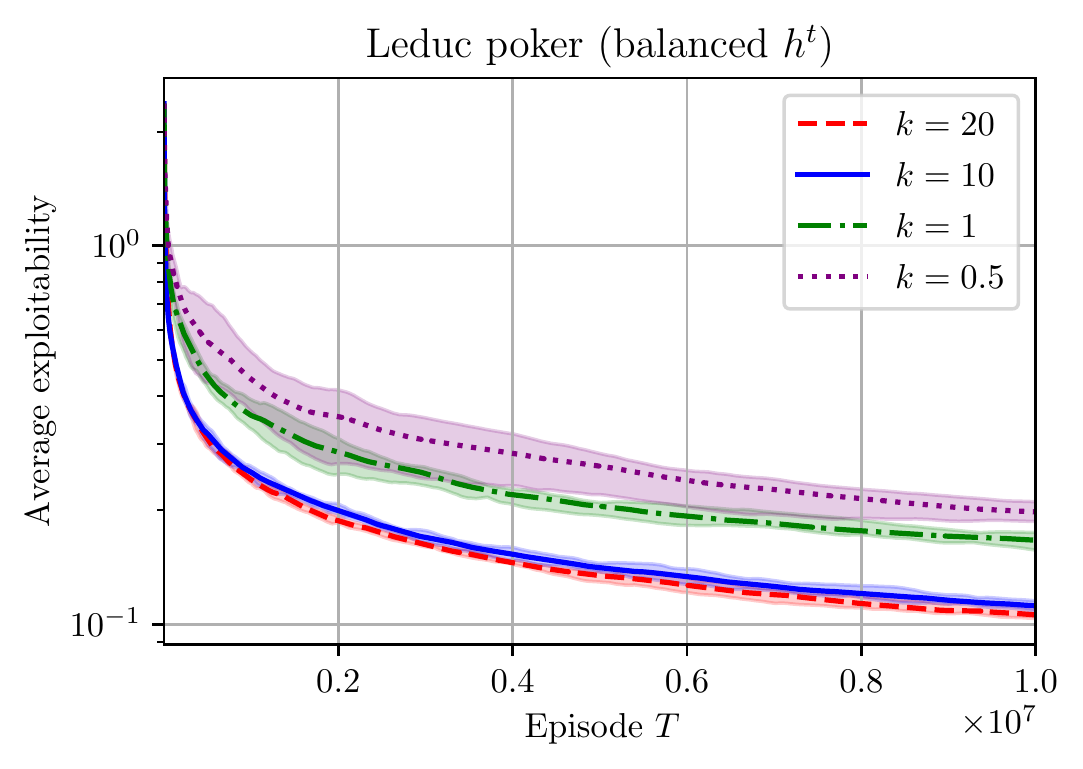}
  \caption{Comparison between uniform and balanced exploration strategies in Leduc poker, for different values of the exploration multipler $k$, when using the on-path-flipping sampling scheme.}\label{fig:leduc conv}
\end{figure}

\begin{figure}[H]
  \centering
  \includegraphics[scale=.65]{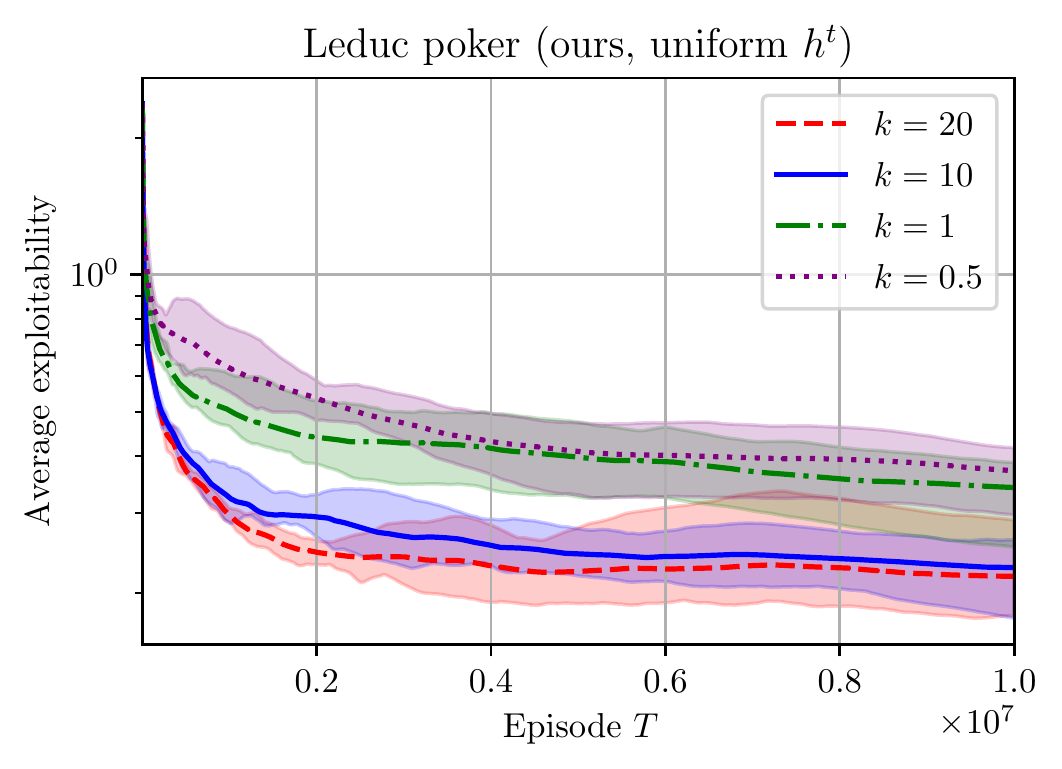}%
  \includegraphics[scale=.65]{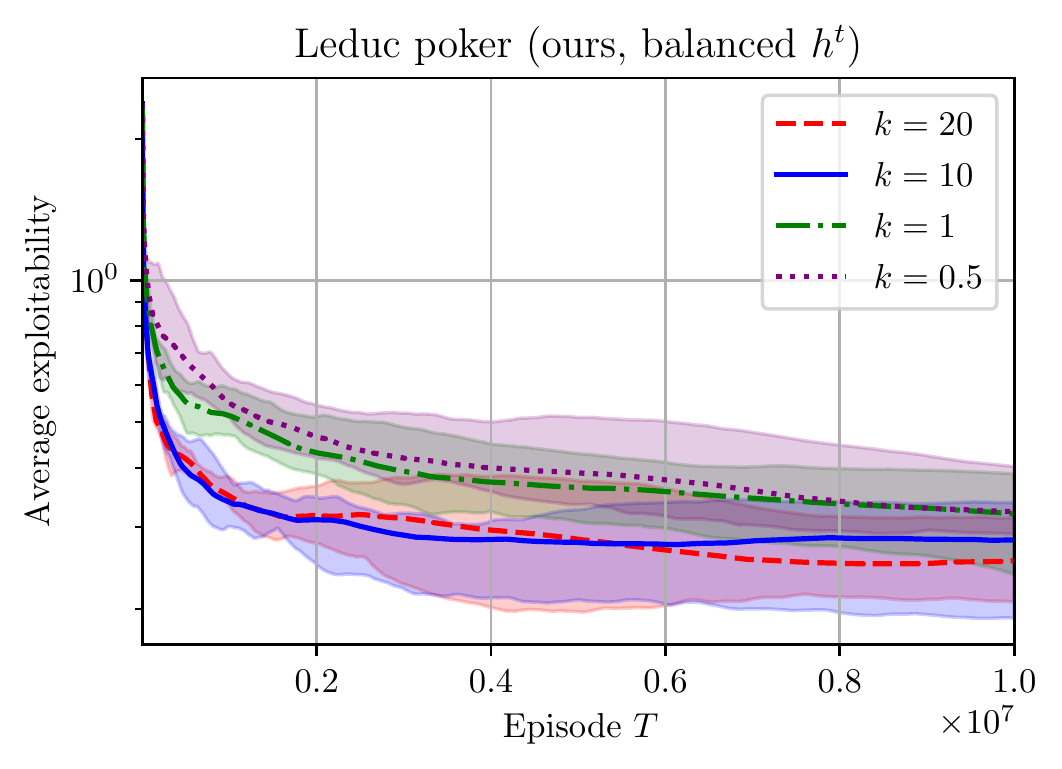}
  \caption{Comparison between uniform and balanced exploration strategies in Leduc poker, for different values of the exploration multipler $k$, when using the upfront-flipping sampling scheme.}\label{fig:leduc coin}
\end{figure}
\begin{figure}[H]
  \centering
  \includegraphics[scale=.65]{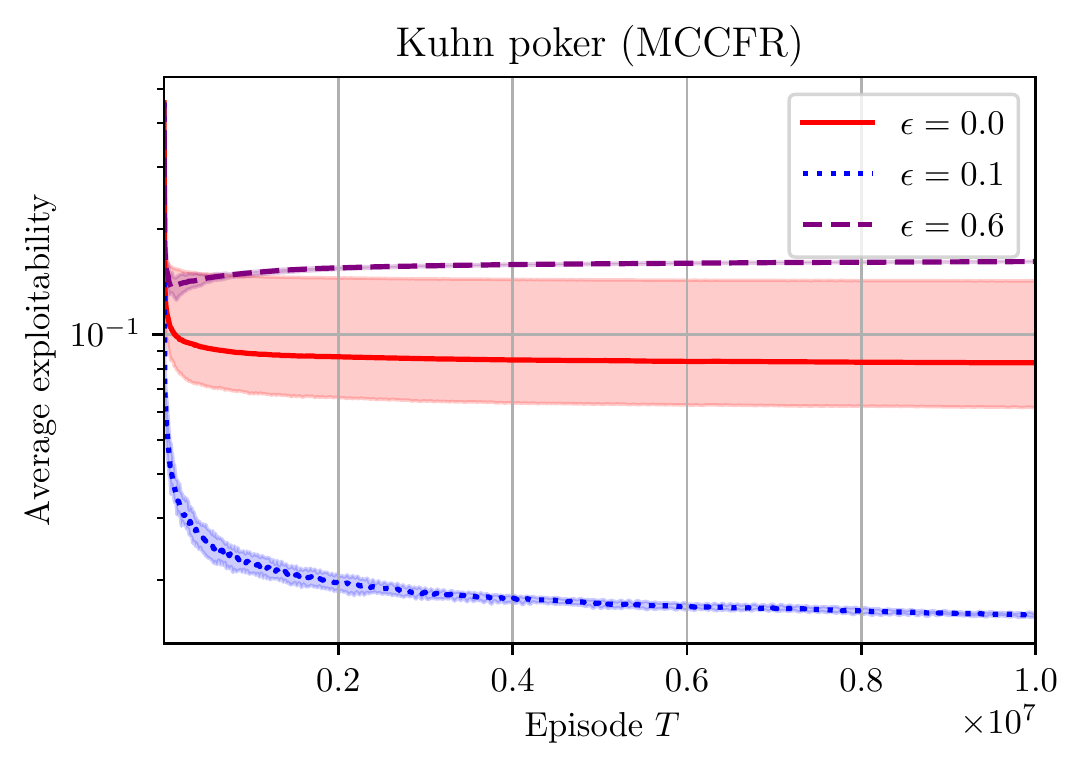}%
  \includegraphics[scale=.65]{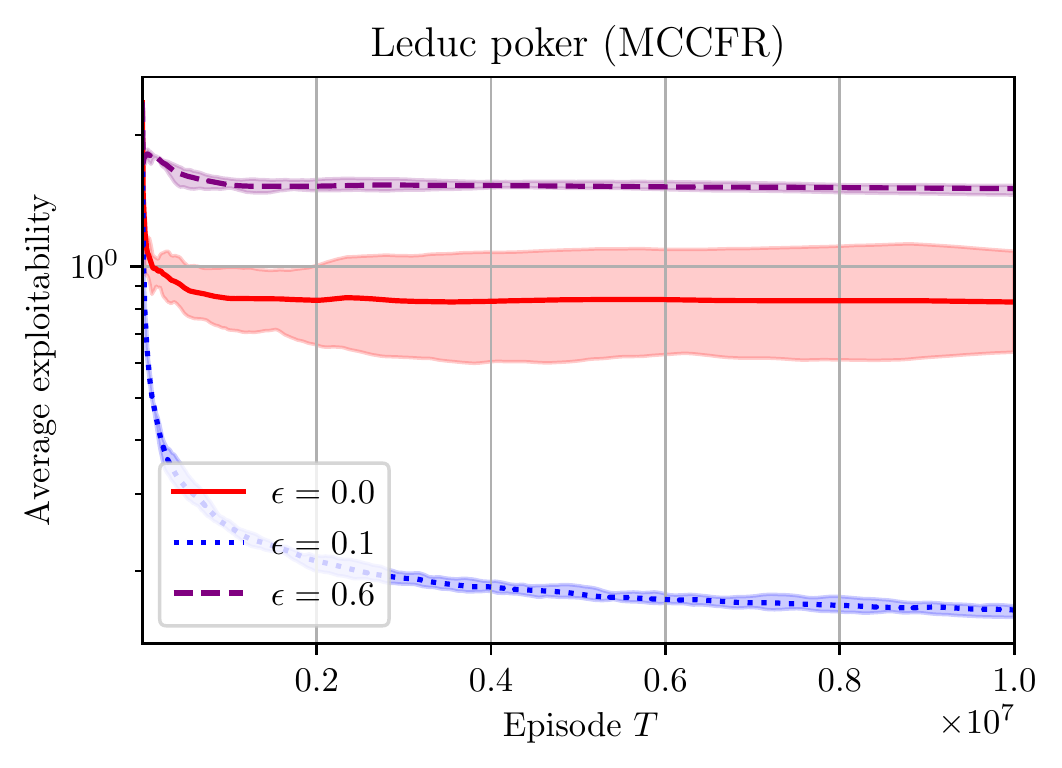}
  \caption{Performance of MCCFR for different amounts of $\epsilon$-greedy exploration.}\label{fig:mccfr}
\end{figure} 

\fi
\end{document}